\RequirePackage{fix-cm}
\documentclass[smallextended]{svjour3}       
\smartqed  
\usepackage[english]{babel}
\usepackage[T1]{fontenc}
\usepackage{amsmath}
\usepackage{amssymb}
\usepackage{epic,eepic}
\usepackage{theorem}
\usepackage{pifont}  
\usepackage{euscript}
\usepackage{calc}
\usepackage[titles]{tocloft}

\usepackage[usenames]{color}          
\usepackage{xcolor}
\definecolor{Brown}{rgb}{0.55,0.0,0.10}
\definecolor{dgreen}{rgb}{0.00,0.56,0.00}

\usepackage{epsf}           
\usepackage{graphicx}       
\usepackage{epsfig}         
\usepackage{pstricks}
\usepackage{pstricks-add}
\usepackage{pst-plot}
\usepackage{dsfont}
\RequirePackage[dvips,colorlinks,hyperindex]{hyperref}
\hypersetup{
linktocpage=true,
citecolor=Brown,
linkcolor=dgreen
}

\theoremheaderfont{\color{black}\normalfont\bfseries}

{\begin{list}{}{%
\settowidth{\labelwidth}{\textrm{#1~}}%
\setlength{\leftmargin}{\labelwidth+\labelsep}}}
{\end{list}}

\newcommand{\Menge}[2]{\bigg\{{#1}~\bigg |~{#2}\bigg\}}

\newcommand{\conv}{\ensuremath{\operatorname{conv}}}

\newcommand{\argmin}{\ensuremath{\operatorname{argmin}}}

\newcommand{\R}{\mathbb{R}}

\newcommand{\N}{\mathbb{N}}

\newcommand{\E}{\mathbb{E}}

\newcommand{\G}{\mathcal{G}}

\newcommand{\C}{\mathcal{C}}

\newcommand{\QQ}{\mathcal{Q}}

\newcommand{\RR}{\mathcal{R}}

\newcommand{\T}{\mathcal{T}}

\newcommand{\PP}{\mathcal{P}}

\newcommand{\mc}{\mathcal}
\newcommand{\bd}{\textbf}
\newcommand{\ds}{\displaystyle}

\begin{document}

\title{Transforming Monitoring Structures with
Resilient Encoders. Application to Repeated Games.
}


\author{Maël Le Treust        \and
        Samson Lasaulce   
}

\authorrunning{Maël Le Treust
\and
Samson Lasaulce } 

\institute{Maël Le Treust  \at
Université Paris-Est Marne-la-Vallée\\
Institut d'électronique et d'informatique Gaspard-Monge\\
e-mail  : mael.letreust@univ-mlv.fr\\
           \and
Samson Lasaulce \at
Laboratoire des Signaux et Systèmes\\
CNRS - Université Paris Sud XI - Supélec\\
e-mail : samson.lasaulce@lss.supelec.fr\\
}

\date{Received: date / Accepted: date}

\maketitle

\begin{abstract}
An important feature of a dynamic game is its monitoring structure namely, what the players effectively see from the played actions. We consider games with arbitrary monitoring structures. One of the purposes of this paper is to know to what extent an encoder, who perfectly observes the played actions and sends a complementary public signal to the players, can establish perfect monitoring for all the players. To reach this goal, the main technical problem to be solved at the encoder is to design a source encoder which compresses the action profile in the most concise manner possible. A special feature of this encoder is that the multi-dimensional signal (namely, the action profiles) to be encoded is assumed to comprise a component whose probability distribution is not known to the encoder and the decoder has a side information (the private signals received by the players when the encoder is off). This new framework appears to be both of game-theoretical and information-theoretical interest. In particular, it is useful for designing certain types of encoders that are resilient to single deviations and provide an equilibrium utility region in the proposed setting; it provides a new type of constraints to compress an information source (i.e., a random variable). Regarding the first aspect, we apply the derived result to the repeated prisoner's dilemma.


\end{abstract}

\section{Introduction}
\label{intro}

The set of equilibrium utilities of a non-cooperative dynamic game is strongly related to the observation capabilities of the players. For instance, in a long-run repeated prisoner's dilemma where the two players do not see anything from the played actions (blind players), the only equilibrium point corresponds to the inefficient outcome where both players defect \cite{Aumann(SurveyRG)81}\cite{sorin-handbook-1992}. On the other hand, when players perfectly observe all the actions which have been played (perfect monitoring assumption), efficient equilibria can be sustained; in particular, the social optimum is an equilibrium point of the infinitely repeated dilemma or its version with discount factor. This special case illustrates the potential need for being able to transform the monitoring structure of a repeated game into a new one. The relevance of such a transformation may appear in other types of settings such as stochastic games, multi-agent learning, or networked optimization. For example, perfect monitoring (PM) can be targeted to implement the standard fictitious play or best-response algorithms (see e.g., \cite{Peyton(book)04}). The desired final monitoring structure (i.e., after transformation) does not necessarily need to be PM and, for example, ensuring that the players observe (thanks to the transformation) a certain public signal can be sufficient to obtain efficient outcomes for the game. The solution proposed in this paper is to implement this monitoring structure transformation by adding an external agent or encoder (whose role is not strategic but only to encode signals and send them to the players to improve their observation capability) to the initial game. For the sake of clarity and simplicity, the encoder is assumed to perfectly observe the actions played and the desired structure, after transformation, is PM. Note that PM at the encoder is not always necessary to ensure PM for the players (see \cite{LetreustLasaulce(GAMECOMM)11}). Interestingly, there exist some practical scenarios where assuming PM at a terminal is relevant. In wireless communications, the decentralized multiple access channel case is known to be very important \cite{letreust-twc-2010}. In this scenario, there are one receiver (e.g., a WiFi access point or a base station) and several transmitters (e.g., mobile terminals) which choose freely their transmission policy (say their power allocation policy) in order to optimize some performance metric such as the individual transmission rate. Considering that the base station has a computational and observation capability much larger than the mobile transmitters is a typical assumption in wireless communications (see e.g., \cite{kowalewski-globecom-2000}\cite{dasilva-vtc-2001}). As a consequence, the receiver can, in particular, have the role of an encoder which sends a feedback on the played actions to the transmitters. Another important scenario of practical interest for which the framework proposed in this paper is fully relevant is the case of sensor networks with a fusion center (see e.g., \cite{akyildiz-cn-2002}).


One of the main issues addressed in this paper is the design of an encoder which is capable of transforming a monitoring structure by sending complementary public signals to the players. The problem comes from the fact that the set of public signals has a fixed cardinality. One of the consequences of this assumption is the existence of an information constraint on the played action profiles and more precisely on their distribution, and therefore on the feasible players' utilities. As explained further, characterizing this information constraint amounts to designing an encoder which represents the information source (namely, the action profile) in a manner as concise as possible. However, to make the source encoder able to operate at equilibrium (and therefore characterize equilibrium utilities), the encoder has to possess a certain property, called the resilience property \cite{LetreustLasaulce(NetGCoop)11}, which has a cost in terms of compression efficiency. In terms of communication, such a property ensures that, even when one player uses a distribution on his action sequences which is arbitrary and unknown to the encoder, PM remains guaranteed. In strategic terms, if we consider the case of repeated games (which is the case study chosen in this paper), it means that grim-trigger-like plans can be implemented.

The paper is structured as follows. A state of the art on the problem under
investigation is done in Sec. \ref{sec:review}. Sec. \ref{sec:info-constraints}
exploits information-theoretic tools to derive one of the central two results
of this paper which is the information constraint (\ref{eq:RobustSide})  stated in Theorem \ref{theo:RobustSide} and
explains how this constraint translates into a set of action profile distributions
(and therefore into feasible utilities) that are compatible with the perfect
monitoring assumption. Sec. \ref{sec:iid-equilibria} provides the second
important result, stated in Theorem \ref{theo:FolkEncodeurAssiste}, which is an achievable equilibrium utility region for
encoder-assisted infinitely repeated games with signals.
The paper is concluded in Sec. \ref{sec:concl}.

\section{Related works}
\label{sec:review}

Before mentioning some relevant works related to the one reported here, it is useful to define at this point a monitoring structure. A monitoring structure is a conditional or transition probability defined by~:
\begin{equation}
\daleth : \mc{A} \longrightarrow \Delta(\mc{S})
\end{equation}
where $\mc{A} = \mc{A}_1 \times \mc{A}_2 \times ...\times \mc{A}_K$ is the discrete set of action profiles, $K$ is the number of players, $\mc{A}_k$ is the discrete set of actions of player $k\in \mc{K}= \{1,2,...,K\}$, $\mc{S} = \mc{S}_1 \times \mc{S}_2 \times ...\times \mc{S}_K$, $\mc{S}_k$ is the discrete set of signals received by player $k$, and the notation $\Delta(\mc{S})$ stands for the set of probability distributions on $\mc{S}$ (unit simplex).

The first relevant body of related works comprises papers providing lossless \cite{shannon-bell-1948} and zero-errors \cite{Shannon-zero-error56}\cite{Witsenhausen76} source coding theorems. Indeed, the role of the encoder in this paper is to encode a sequence (or block) of  action profiles into a sequence of public signals which is observed by the players. As already mentioned, making this in a concise manner is of prime interest to characterize the information constraint. The considered source coding problem has two main features~: the decoders (namely, the players) have a side information on the source (the private signal) and we want the encoder to be resilient to single deviations that is, the past action profiles are decoded reliably even when the probability distribution of the action of a given player varies arbitrarily over time and remains unknown to both the encoder and decoders. Remarkably, the information theory literature provides the right framework to design such encoders. The corresponding framework is the one of arbitrary varying sources (AVS)~: the source distribution $\PP_v(a)\in \Delta(\mc{A})$ can vary from sample to sample, depending on a parameter or state $v\in\mc{V}$ which represents, in our setting, the probability distribution of the deviator's action. The most relevant works on AVS is based on graph coloring \cite{BondyMurty} and can be found in \cite{Ahlswede(ColoringPart1)79} and \cite{Ahlswede(ColoringPart2)80}. Indeed, the latter references deal with the scenario of two correlated sources either in the case where the destination is informed with the sequence of states or in the case where it is not known. The work reported in this paper is precisely related to the scenario of two arbitrary varying correlated sources of actions $\boldsymbol{a}$ and private signals $\boldsymbol{s}_k$ with a destination (i.e. player $k\in\mc{K}$) uninformed of the state (i.e. strategy of an eventual deviator)~; this scenario is described by Fig. \ref{fig:ReconstructionRobustSideCapaK}. One of our contributions, in addition to establishing a link between equilibrium utility regions and the AVS literature, is to show that the entropy positiveness condition (EPC) in \cite{Ahlswede(ColoringPart1)79}\cite{Ahlswede(ColoringPart2)80}, under which source coding rates (i.e. optimal compression level) can be characterized, can be removed and replaced with another mathematical condition which is of strong game-theoretic interest namely, the resilience property. Additionally, it holds for some useful special cases for which the EPC is not met, the case of deterministic channels in particular.

The second body of works concerns works on folk theorems. The stronger results have been obtained for one of the simplest classes of dynamic games namely, the one of repeated games (see e.g., \cite{sorin-handbook-1992} for a survey). The standard approach consists in assuming a given monitoring structure (e.g., standard-trivial monitoring \cite{Lehrer91}, public monitoring \cite{FudenbergLevineMaskin94}, or almost-perfect monitoring \cite{HornerOlszewski06}) and, then, deriving a folk theorem. Compared to these works, our approach is different since we do not try to characterize the equilibrium utilities of a repeated game with an arbitrary monitoring structure (which is an open problem \cite{RenaultTomala(GeneralProp)11}). Rather, our approach aims at transforming, with an additional encoder, an arbitrary monitoring structure of any dynamic game into a new monitoring structure for which the equilibrium utilities can be fully characterized~; in this paper, PM is the targeted final structure. Even though the final monitoring structure is PM, there are still some differences between a dynamic game with PM (the focus will be on repeated games here) and a dynamic game where players have PM thanks to the encoder~:
\begin{itemize}
  \item there exists an internal information constraint on the action distribution due to the fact that
      the set of public signals has a fixed cardinality~;

  \item action profiles are encoded by blocks by the encoder and each player decodes a block of played actions from a whole block of observations. Therefore PM is established with a delay~;

  \item only i.i.d equilibrium utilities (and convex combinations of them) are studied. This assumption
on the action profiles is well motivated in the paper and does not prevent us from deriving useful results which may be extended if needed.
\end{itemize}
For all of these reasons, we will use the term ``virtually perfect monitoring'' (VPM) to refer to such a framework.

To conclude on the most relevant references related to the work reported in this paper, we will mention a couple of references at the intersection between game and information theory. For instance, in \cite{lehrer-jet-1988}, \cite{NeymanBavly03}, \cite{Peretz11} entropy-based information constraints are used to characterize the individually rational levels of repeated games with bounded recall. In \cite{GossnerTomala07}, the authors characterize the maximum utility a team can guarantee against another in a class of repeated games with imperfect monitoring by exploiting a constraint on possible correlation schemes expressed in terms of entropy variation. In \cite{GossnerHernandezNeyman06}, the authors are exploiting an information constraint in the sense of the present work that is, the source coding rate has to be less than the channel capacity, although the constraint is not interpreted this way in their work. This leads to a characterization of equilibrium utilities a team of two players can implement when only one player is (noncausally) informed of the i.i.d. sequence of states of the repeated game.

\section{Virtual perfect monitoring of an arbitrarily varying information source}
\label{sec:info-constraints}

\subsection{Methodology}
\label{sec:subsec-methodology}

The scenario under consideration is as follows (see Figure \ref{fig:ReconstructionRobustSideCapaK}).
Let us fix a family of probability distributions $\PP_k^{\star} \in \Delta(\mc{A}_k)$ with $k\in\mc{K}$.
When a given action profile $a=(a_1,a_2,...,a_K)\in \mc{A}$ is drawn from the product probability
$\PP^{\star} = \PP^{\star}_1 \otimes \ldots \otimes \PP^{\star}_K \in \Delta(\mc{A})$, player $k\in\mc{K}$ receives a symbol
$s_k \in \mc{S}_k$ with a probability given
by the conditional probability
\begin{equation}
\daleth(s_k | a) = \sum_{s_{-k}\in\mc{S}_{-k}} \daleth(s_k, s_{-k} | a).
\end{equation}
An encoder $\C$, who perfectly monitors the played actions, encodes the
observed action profiles by blocks or sequences of size $n\geq1$ into a sequence
of public signals $s_0 \in \mc{S}_0$ which are received by all the players. These
public signals form a perfect channel of capacity $ \log_2|\mc{S}_0|$,
which is orthogonal to the one defined by $\daleth$ that is, player $k$
actually receives a pair of signals $(s_k, s_0)$ for every played action profile.
Note that player $k$ recall it's own action $a_k$.
The purpose of the encoder is to use the minimal amount of additional information,
in order for every player to acquire the information which is missing to have PM.
In what follows, we first define a code in our setup.
Second, we define the notion of virtually perfect monitoring (VPM) of actions profile
$a=(a_1,a_2,...,a_K)\in \mc{A}$ defined as an arbitrarily varying information source (AVS).
 Third, we prove a theorem which state an information constraint on the action
profile distribution which is due to the fact that the
communication channel between the encoder and players has a limited capacity.
Denote $\mc{A}^n$ (resp. $\mc{A}^{\infty}$) the set of sequences $a^n\in\mc{A}^n$ of length $n\in\N$ (resp.
of sequences $a^{\infty}\in\mc{A}^{\infty}$ of infinite length).

 \begin{figure}[h!]\label{fig:ReconstructionRobustSideCapaK}
\begin{center}
\psset{xunit=0.7cm,yunit=0.7cm}
\begin{pspicture}(-4,-1)(16,8)
\psframe(-3,8)(-2,7)
\psframe(-3,6)(-2,5)
\psframe(-3,2)(-2,1)
\rput[u](-2.5,7.5){$\#_1$}
\rput[u](-2.5,5.5){$\#_2$}
\rput[u](-2.5,1.5){$\#_K$}
\psdots(-2.5,4)(-2.5,3.5)(-2.5,3)
\psline[linewidth=1pt]{->}(-2.05,7.05)(0.3,4.75)
\psline[linewidth=1pt]{->}(-2.05,5.05)(0.3,4.65)
\psline[linewidth=1pt]{->}(-2.05,1.95)(0.28,4.3)
\rput[u](-1,6.5){$\bd{a}_1$}
\rput[u](-1,5.2){$\bd{a}_2$}
\rput[u](-1,2.2){$\bd{a}_K$}
\rput[u](2.2,4.9){$ (\bd{a}_1, \bd{a}_2,\ldots,\bd{a}_K)$}
\psdots(0.5,4.5)(2.5,4.5)
\psline[linewidth=1pt]{->}(0.5,4.5)(4,4.5)
\psline[linewidth=1pt](2.5,4.5)(2.5,0)
\psline[linewidth=1pt](2.5,0)(10.5,0)
\psline[linewidth=1pt]{->}(10.5,0)(10.5,4)
\pscircle(4.5,4.5){0.35}
\rput[u](4.5,4.5){$\daleth$}
\psline[linewidth=1pt]{->}(4.9,4.8)(7,7)
\psline[linewidth=1pt]{->}(5,4.5)(7,5)
\psline[linewidth=1pt]{->}(4.8,4.15)(7,2)
\rput[u](6,6.6){$\bd{s}_1$}
\rput[u](6,5.1){$\bd{s}_2$}
\rput[u](6,2.2){$\bd{s}_K$}
\rput[u](6.8,8.1){$\bd{a}_1$}
\rput[u](6.8,6.1){$\bd{a}_2$}
\rput[u](6.8,0.9){$\bd{a}_K$}
\psframe(7,8)(8,7)
\psframe(7,6)(8,5)
\psframe(7,2)(8,1)
\rput[u](7.5,7.5){$\#_1$}
\rput[u](7.5,5.5){$\#_2$}
\rput[u](7.5,1.5){$\#_K$}
\psdots(7.5,4)(7.5,3.5)(7.5,3)
\psframe(10,4)(11,5)
\rput[u](10.9,3.5){$\bd{a}$}
\rput[u](10.5,4.5){$\C$}
\psline[linewidth=1pt]{->}(10,4.8)(8,7)
\psline[linewidth=1pt]{->}(10,4.5)(8,5)
\psline[linewidth=1pt]{->}(10,4.15)(8,2)
\rput[u](9,6.5){$\bd{s}_0$}
\rput[u](9,5.2){$\bd{s}_0$}
\rput[u](9,2.2){$\bd{s}_0$}
\end{pspicture}
\caption{Each action profile of the game $\bd{a} = (\bd{a}_1,\bd{a}_2,\ldots,\bd{a}_K)$ generates a signal profile $(\bd{s}_1,\bd{s}_2,...\bd{s}_K)$ through a condition probability $\daleth$. Player $\#k$ (represented twice here above) only observes $\bd{s}_k$ from this action profile $\bd{a}$. The encoder $\C$, who perfectly monitors the played action profiles $\bd{a}$, builds a complementary public signal $\bd{s}_0$ which is observed by all the players. Each player has to reconstruct virtual perfect monitoring (VPM) from a sequence of pairs of signals $(\bd{s}_k,\bd{s}_0)$ and the knowledge of the sequence of its individual actions $\bd{a}_k$.}
\end{center}
\end{figure}
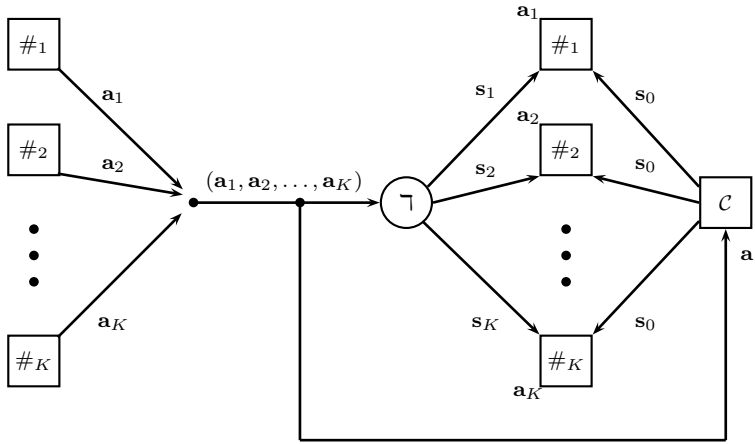

\subsection{Information constraint for resilient coding with side information at the decoder}
\label{sec:subsec-resilient-source-coding}

Here, we assume that the distribution of the source may vary
from stage (or action profile) to stage (or action profile)~; this is the framework of arbitrarily varying source (AVS) coding.

\begin{definition}[Arbitrarily Varying Source (AVS)]\label{def:AVSunilateral}
Let $\PP^{\star}\in \Delta(\mc{A}) $ a probability distribution
(mixed strategy) and $ \mc{V}$ the set of states of the source:
\begin{eqnarray}
\mc{V} = \ds{\cup_{k\in \mc{K}}} \Delta(\mc{A}_k^{ \infty}).\label{Condition:SuiteDEtats}
 \end{eqnarray}
The arbitrarily varying (AVS) information source $\bd{a}\in \mc{A}$ is at a certain
state $v \in \mc{V}$, when one component of the action
profile has a distribution which may vary arbitrarily over time and is fully
unknown to the coder. For example, when the sequence of the states is
$v = \QQ_i \in \Delta(\mc{A}_i^{ \infty}) \subset \mc{V}$, the sequence of
actions $\bd{a}^n=(\bd{a}_1^n,\ldots,\bd{a}_K^n)$ is drawn following a
probability distribution given by~:
 \begin{eqnarray}
\PP_{v}(\bd{a}_1^n,\ldots,\bd{a}_i^n,\ldots,\bd{a}_K^n)
= \bigg[\PP_1^{\star \;\otimes n}\otimes\ldots\otimes \QQ_i\otimes \ldots\otimes \PP_K^{\star \;\otimes n}\bigg](\bd{a}^n).\label{eq:DistribArbitrary}
\end{eqnarray}
\end{definition}
Now, we formally define the notion of code for the AVS represented by Fig. \ref{fig:ReconstructionRobustSideCapaK}.
\begin{definition}
A code $\lambda$ of size $n$ for the encoder $\C$ and decoders $\mc{K}$ consists of an encoding function $f_0$ and $K$ decoding functions $(g_k)_{k\in\mc{K}}$ defined as~:
\begin{equation}
\label{eq:Encoding}
\begin{array}{cccccc}
f_0 & : & \mc{A}^n   & \longrightarrow & \mc{S}_0^n&\\
g_k & : & \mc{S}_0^n  \times \mc{S}_k^{n} \times \mc{A}_k^n & \longrightarrow &\mc{A}^n,&\qquad \forall k\in\mc{K}
\end{array}.
\end{equation}
Denote by $\Lambda(n)$, the set of codes for which the length $n\in \N$ of the code-words is fixed.
\begin{eqnarray}
&\PP_e(\lambda) =& \sum_{k\in\mc{K}} \max_{i \in \mc{K}}\max_{v_i \in  \Delta(\mc{A}_i^{\infty})}
\PP_{v_i}(\bd{a}^n \neq g_k(\bd{s}_k^n,\bd{s}_0^n,\bd{a}_k^n)),\label{eq:ErrorProbaK}
\end{eqnarray}
The error probability $\PP_e(\lambda)$ of the code $\lambda\in\Lambda(n)$ is defined by equation
(\ref{eq:ErrorProbaK}) and corresponds to the sum of the error probability for each decoder
$k\in\mc{K}$, considering every possible deviation $v_i \in  \Delta(\mc{A}_i^{\infty})$ of
player $i \in \mc{K}$ (i.e. any variation of the source).
\end{definition}

\begin{definition}[Virtually Perfect Monitoring (VPM)]
Players $\mc{K}$ have a virtually perfect monitoring (VPM) of the information
source $\bd{a}\in \mc{A}$ if for all $\varepsilon>0$,
there exists a parameter $n\in \N$, and a code $\lambda\in \Lambda(n)$ such that:
\begin{eqnarray}
\PP_e(\lambda) &\leq& \varepsilon,\label{eq:cond:ErrorProba}
\end{eqnarray}
\end{definition}

The condition (\ref{eq:cond:ErrorProba}) means that it is possible to find coding and
decoding functions to represent any sequence of $n$ realizations of the
$K-$dimensional random variable $\boldsymbol{a}$ with $2^{n\log_2|\mc{S}_0|}$
indices or sequences of public signals in such a way that, any decoder $k$, based on the
knowledge of $(s_0^n,s_k^n,a_k^n)$, can find the sequence $a^n$ with an
arbitrarily small probability of error. In a game theoretical framework,
the players virtually perfect monitor the sequences of past actions played.

At this point, the main issue is to be able to characterize the set of
AVS information source that are compatible with
the VPM of the players $\mc{K}$.
The AVS hypothesis guarantee that the past actions
played will be observed by all the players even if one of them deviates,
 manipulates the coding scheme $\lambda\in \Lambda(n)$ in order to break reliability.
Theorem \ref{theo:RobustSide} provides an information constraint which guarantee
VPM for the AVS information source of player's actions $\bd{a}$.
To state this theorem, an auxiliary graph needs to be defined first.

\begin{definition}[Auxiliary graph]\label{def:graph} For each player $i\in \mc{K}$, an auxiliary graph $\G_i$ is defined as follows $\G_i=(\mc{A}_i, \mc{E}_i)$. The actions $a_i\in \mc{A}_i$ of player $i\in \mc{K}$ are the vertices of the graph. There exists an edge $e_i=(a_i,a_i')\in \mc{E}_i$ between two actions $a_i \in \mc{A}_i$ and $a'_i \in \mc{A}_i $ if~:
\begin{eqnarray*}
 & \exists\; a_{-i} \in  \textrm{Supp}\;\PP^{\star}_{-i},\; \exists k\in\mc{K},\; \exists s_k\in \mc{S}_k,\; \exists \delta>0,\text{ s.t. } &\\ &\min(\daleth(s_k|a_i,a_{-i}), \daleth(s_k|a_i',a_{-i}))\geq\delta&
\end{eqnarray*}
where $\textrm{Supp}\; \PP^{\star}_{-i}$ is the support of the probability distribution $\PP^{\star}_{-i}$ defined by $ \PP^{\star}_{-i} = \bigotimes_{j \neq i } \PP^{\star}_j \in \prod_{j\neq i} \Delta(\mc{A}_j)$.
\end{definition}
Two vertices $a_i \in \mc{A}_i$ and $a'_i \in \mc{A}_i $ are neighbors in the graph $\G_i$ if the probability that these actions lead, through $\daleth$, to the same signal $s_k\in \mc{S}_k$ for at least one player $k\in\mc{K}$ is not zero.
Now, to define the chromatic number \cite{BondyMurty} of the graph $\G_i$, we define the notion of coloring in our context.

\begin{definition}[Coloring]\label{def:Coloring} Let $\Phi_i$ a set of colors. A coloring of the graph $\G_i$ is a function $\phi_i: \mc{A}_i \longrightarrow \Phi_i$ which satisfies~:
\begin{equation}
\forall e_i=(a_i,a'_i)\in \mc{E}_i,\text{ we have that  } \phi_i(a_i)\neq \phi_i(a'_i).
\end{equation}
\end{definition}

A minimal coloring of the graph $\G_i$ is a coloring $\phi_i$ for which the cardinality of the set of colors $\Phi_i$ is minimal. The chromatic number $\chi_i$ of the graph $\G_i$ is the cardinality $|\Phi_i|$ of the set of colors of the minimal coloring of the graph $\G_i$. This is precisely this quantity which is used in the next theorem.

\begin{theorem}[Coding result for AVS]\label{theo:RobustSide}
Players $\mc{K}$ have a virtually perfect monitoring (VPM) of the arbitrarily varying (AVS)
information source $\bd{a}\in \mc{A}$ if the following condition is met~:
\begin{equation}
\textsf{R}^{\star} = \max_{i\in\mc{K}} \bigg[\max_{k\in\mc{K},\atop a_i\in \mc{A}_i}H(\bd{a}_{-i,k}|\bd{s}_k(a_i), \bd{a}_k)+\log_2 \chi_i \bigg] < \log_2 |\mc{S}_0|, \label{eq:RobustSide}
\end{equation}
where~:
\begin{itemize}
  \item[$\bullet$] $\bd{a}_{-i,k}$ is the action profile without the components $i$ and $k$. It is distributed as $\PP^{\star}_{-i,k}\in \prod_{j\neq i,\atop j\neq k } \Delta(\mc{A}_j)$~;

  \item[$\bullet$] $\bd{s}_k(a_i)$ is the signal received by player $k$ when the action $a_i$ is fixed. It is induced by $\bd{a}_{-i}$ and the transition probability $\daleth$~:
  \begin{eqnarray}
\daleth_{a_i} : \mc{A}_{-i} \longrightarrow& \Delta(\mc{S}_k)&\\
a_{-i}\longrightarrow&\daleth_{a_i}(s_k|a_{-i})&=\daleth(s_k|a_i,a_{-i}) \nonumber\\
&&= \sum_{s_{-k}\in\mc{S}_{-k}}\daleth(s_k,s_{-k}|a_i,a_{-i})~;
\end{eqnarray}

\item[$\bullet$] $\log_2|\mc{S}_0|$ is given by the cardinality of the set of public signals and corresponds to the capacity of the perfect channel between the encoder $\C$ and the players $\mc{K}$.
\end{itemize}
\end{theorem}

Several comments are in order. First, let us comment on the main assumptions. The i.i.d assumption over time made on the source to be encoded is common in the information theory literature and will only be briefly commented. Solving the i.i.d. case might not only be helpful but even sufficient for solving the case with arbitrary correlation between consecutive source samples. To be more specific, if the source generates $B$ blocks of $\ell$ correlated symbols for $B$ sufficiently large, $\ell < +\infty$, and i.i.d. blocks, then the information constraint directly follows from the original i.i.d case (concerning i.i.d symbols) by considering vectors of symbols instead of symbols. Beyond this framework, the source coding literature comprises works dealing with refinements such as universal coding \cite{Gallager(AVS)76} and information-spectrum based coding \cite{Han(InfoSpectrumBook)06}.
Now, from a game-theoretic perspective, studying sequences of i.i.d profiles (up to one component) is not only an intermediate case which can be challenging technically (think of repeated games with arbitrary monitoring structures) but also to design implementable equilibrium action plans. As for relaxing the i.i.d assumption over space (over the components), provided the resilience property is relaxed and the joint distribution on the actions is known to the encoder, it only consists in changing scalar quantities into vectors (of size $K$). When resilience to single deviations is required, the spatial i.i.d assumption is useful to derive information constraint (as advocated by the proof provided in App. \ref{sec:DemoTheoRobustSide}) but studying necessity is a possible extension of this work. At last note that the spatial i.i.d assumption allows one to study mixed strategies which is known to be important.

Now, let us comment on the result i.e., the information constraint defined by (\ref{eq:RobustSide}). The presence of the maximum over $i$ is due to the fact that the location of the component (which corresponds to the deviator in a game), whose distribution is unknown, is itself unknown to $\C$. The second maximum over $k$ and $a_i$ indicates the case where the deviator $i$ chooses the worst action $a_i$ in terms of coding efficiency for to the worst decoder $k$. The conditioning w.r.t. $(\bd{s}_k(a_i), \bd{a}_k)$ in the entropy translates the knowledge of the decoder in terms of side information, which therefore reduces the entropy. The isolated term $\log_2 \chi_i$ corresponds to the amount of information needed by $\C$ to encode a component separately~; since the probability distribution of $\bd{a}_i$ is unknown, symbol-by-symbol coding is optimal here. Without side information at the decoder $i$ this quantity would be $\log_2 |\mc{A}_i|$. At last, the righthandside term $\log_2|\mc{S}_0|$ corresponds to the channel capacity of a broadcast channel with a public message and for which the decoders directly observe the signal sent by the encoder (see e.g., \cite{cover-book-2006}).

To conclude this section, let us comment on the proof of this theorem. Although the detailed proof of this theorem is provided in Sec. \ref{sec:DemoTheoRobustSide}, we would like to mention here some technical differences w.r.t the derivation made by Ahlswede in \cite{Ahlswede(ColoringPart2)80}. The imposed condition is totally different. Imposing resilience to single deviations to the source encoder requires to transmit without error the sequence of actions of the deviating player. Our proof is based on a sequence of coloring where the vertices are the symbols whereas Ahlswede \cite{Ahlswede(ColoringPart2)80} use a coloring where the vertices are the sequences of symbols. To exploit the law of large numbers for the sequences of symbols, his proof requires an additional condition which is EPC. In our framework, this condition is removed and replaced with a condition over the admissible sequences of states (\ref{Condition:SuiteDEtats}) and by the feature that the random signal $\bd{s}$ depends on the state $v$ only through the action $\bd{a}$. Our result is applicable to the case of deterministic transition probability $\daleth$ whereas this special type of transition probabilities does not meet EPC.

\section{Equilibrium utilities of an encoder-assisted repeated games with signals}
\label{sec:iid-equilibria}

The goal of this section is to characterize equilibrium utilities of an infinite repeated game with signals where an additional encoder establishes VPM.
To this end, notations, definitions and results of the preceding section are used.

\subsection{Game formulation and main result}
\label{sec:game-model}

We consider an encoder-assisted repeated game with signals. The stage or constituent game is given by the triplet $\left(\mc{K}, (\mc{A}_k)_{k\in \mc{K}}, (u_k)_{k\in \mc{K}} \right)$, where $u_k \in \R$ is the utility function of player $k\in \mc{K}$. The private monitoring  structure is given by the conditional probability $\daleth(s|a): \mc{A} \longrightarrow \Delta(\mc{S})$. The encoder $\C$ is assumed to perfectly monitor the past action profile $a\in\mc{A}$ and send a public message $s_0 \in\mc{S}_0$ to the players.\\
A strategy for the encoder (by abuse of language we use the term strategy here even though in this paper the encoder has no utility in the game-theoretic sense) is a  sequence of causal functions or mappings  $\sigma =  (\sigma^t)_{t\geq 1}$ with $\forall t\geq1,\; \quad \sigma^t : \mc{A}^{t-1} \times \mc{S}_0^{t-1} \rightarrow \mc{S}_0$~; $t$ stands for the stage index and $a_t$ is the profile played at stage $t$~; the set of strategies of the encoder will be denoted by $\Sigma$.\\
A behavior strategy for a player is a  sequence of causal functions or mappings $(\tau_k^t)_{t\geq 1}$ with $\forall t\geq1,\; \quad\tau_k^t : (\mc{A}_k \times \mc{S}_k \times \mc{S}_0)^{t-1}  \rightarrow  \Delta(\mc{A}_k)$~; the notation $\tau=(\tau_1,\tau_2,...\tau_K)$ will stand for a profile of behavior strategies for the repeated game~; the set of behavior strategies will be denoted by $\T = \prod_{k\in\mc{K}}\T_k $. \\
At last, we will denote by $\PP_{\sigma, \tau}$ the probability distribution on the infinite sequences of actions, private and public signals $((a_k^{\infty})_{k\in\mc{K}},(s_k^{\infty})_{k\in\mc{K}},s_0^{\infty}) \in \mc{A}^{\infty}\times \mc{S}^{\infty} \times \mc{S}_0^{\infty} $ induced by the pair of strategies $(\sigma, \tau) \in \Sigma \times \T$. At this point, one can define a uniform equilibrium of the encoder-assisted repeated game with signal.

\begin{definition}[Equilibrium points]\label{def:EqInfini} A pair of strategies $(\sigma, \tau)\in \Sigma \times \T $ of the encoder $\mc{C}$ and the players $\mc{K}$ is a uniform equilibrium of the encoder-assisted repeated game with signals if~:\\\\
 (i) For each player $k\in \mc{K}$, the expected utility,
 \begin{equation}
 \gamma_k^T(\sigma, \tau) = \E_{\sigma, \tau} \left(\frac{1}{T} \sum_{t=1}^T u_k(a_t)\right),\label{eq:UtilInfini}
 \end{equation}
 has a limit when $T \rightarrow +\infty$; \\\\
 (ii) $\forall \varepsilon >0, \; \exists \bar{T}> 0,\; \forall T\geq  \bar{T}, \; \forall k\in \mc{K},\; \forall \tau'_k \in \T_k  $, such that,
  \begin{equation}
  \gamma_k^T(\sigma, \tau'_k, \tau_{-k}) \leq  \gamma_k^T(\sigma,  \tau_k, \tau_{-k}) + \varepsilon.\label{eq:UtilEqCond}
   \end{equation}
The point $U^{\star}=(U_1^{\star},U_2^{\star},...,U_K^{\star}) \in \R^K$ is a vector of equilibrium utilities if there exists a pair of strategies $(\sigma^{\star}, \tau^{\star})$ such that:
\begin{equation}
\forall k \in \mc{K}, \ \lim_{T\rightarrow +\infty}  \gamma_k^T(\sigma^{\star}, \tau^{\star}) = U_k^{\star}.
\end{equation}
\end{definition}
The set of the equilibrium points of the encoder-assisted repeated game with signals will be denoted by $NE^{\infty}_{enc}$.

\begin{definition}[Individually rational points]\label{def:minmaxIRreal}
The independent min-max level $\upsilon_k$ of player $k\in \mc{K}$ is defined by (\ref{eq:NiveauMinmax}) and is also called punishment or defense level. The individually rational $IR$ utilities are defined by (\ref{eq:IndivRat}) and correspond to the utilities that Pareto-dominate the min-max levels defined as~
\begin{eqnarray}
\upsilon_k&=&\min_{\PP_{-k}\in \prod_{i\neq j} \Delta (\mc{A}_j)} \max_{\PP_k\in \Delta(\mc{A}_k)} \E_{\PP_k,\PP_{-k}}\bigg[u_k(\bd{a}_k,\bd{a}_{-k})\bigg],\quad k\in \mc{K},\label{eq:NiveauMinmax} \\
IR&=& \Menge{ (x_k)_{k\in \mc{K}} \in\R^K}{ x_k\geq \upsilon_k \quad \forall k\in \mc{K} }.\label{eq:IndivRat}
\end{eqnarray}
\end{definition}

\begin{definition}[Information constraint set] The set $\RR$ of mixed actions that satisfy the information constraint
(\ref{eq:RobustSide})  is defined by~:
\begin{eqnarray}
\RR=\Menge{ \PP\in \prod_{k\in \mc{K}}\Delta(\mc{A}_k) }{ \max_{i\in\mc{K}} \bigg[\max_{k\in\mc{K},\atop a_i\in \mc{A}_i}H(\bd{a}_{-i,k}|\bd{s}_k(a_i), \bd{a}_k)+\log_2 \chi_i \bigg]< \log_2|\mc{S}_0| }.\nonumber\\ \label{eq:CondDefMixte}
\end{eqnarray}
\end{definition}

\begin{theorem}[Folk theorem with VPM]\label{theo:FolkEncodeurAssiste}
The  set of utilities $\conv u (\RR)\cap IR$ is included in the set of uniform
equilibrium utilities for the encoder-assisted repeated game with signals~:
 \begin{eqnarray}
\conv u(\RR) \cap IR  \qquad &\subset& \qquad NE^{\infty}_{enc}.
\end{eqnarray}
Moreover, for any utility vector in this set, VPM can be implemented by the encoder $\C$ and the players $\mc{K}$.
\end{theorem}

The proof is provided in Sec. \ref{sec:DemoJeu} and is based on Theorem \ref{theo:RobustSide}. The framework of AVS is exploited to characterize the communication possibilities for the encoder. The first feature of the problem is that the coding scheme must be reliable even if one of the players deviates. The second main feature is that the coding scheme must also take into account the private signals received by the players. These two hypotheses allow us to determine the amount of additional information needed from the encoder in order to
implement VPM. Interestingly, the proof of Theorem \ref{theo:FolkEncodeurAssiste} relies on classical grim-trigger strategies but implemented in a blockwise manner and by exploiting VPM and strong typicality \cite{cover-book-2006} as a statistical test whose result indicates to every player whether to keep on following the main plan.

\subsection{Application to the repeated prisoner's dilemma}\label{sec:repetead-prisoners-dilemma}

We consider a prisoner's dilemma whose matrix form is given by Tab. \ref{tab:prisoners-dilemma}. Let $|\mc{A}|=4$ and $|\mc{S}_0|= 3$. Note that the encoder cannot send the action profile profile directly to the players. The goal of this section is to describe the mixed strategies
$\PP^{\star}\in\Delta(\mc{A})$ that are compatible with the information constraints (\ref{eq:RobustSide}).
\begin{table}[!h]
\begin{center}
\begin{tabular}{|c|c|c|}
  \hline
      &       \textsf{L}      &     \textsf{R} \\
   \hline
  \textsf{T}   &  (3, 3)      & (0, 4) \\
    \hline
  \textsf{B}  &  (4, 0)      & (1, 1) \\
  \hline
\end{tabular}
\caption{The prisoner's dilemma in a matrix form.}\label{tab:prisoners-dilemma}
\end{center}
\end{table}
If this constraint is satisfied, the encoder can compress the sequence of past actions, encode it
into a sequence of public signals and the players can decode the sequence of past action
with an error probability that goes to zero when the length of the sequences goes to infinity.
Denote $\mc{A}_1 = \{T,B\}$ and $\mc{A}_2=\{L,R\}$.

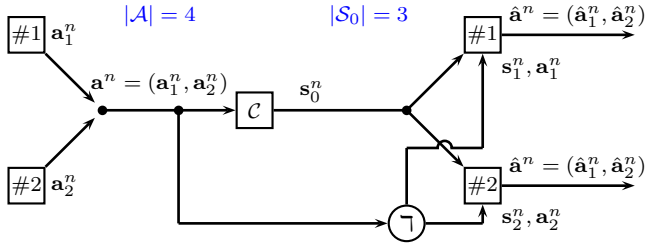
\begin{figure}[h!]
\begin{center}
\psset{xunit=0.5cm,yunit=0.5cm}
\begin{pspicture}(-4,1)(16,7)
\psframe(-2,2)(-1,3)
\psframe(-2,6)(-1,7)
\psframe(4,4)(5,5)
\psframe(10,2)(11,3)
\psframe(10,6)(11,7)
\pscircle(8.5,1.5){0.25}
\psdots(0.5,4.5)(2.5,4.5)(8.5,4.5)
\psline[linewidth=1pt]{->}(-1,3)(0.3,4.3)
\psline[linewidth=1pt]{->}(-1,6)(0.3,4.7)
\psline[linewidth=1pt]{->}(0.5,4.5)(4,4.5)
\psline[linewidth=1pt](5,4.5)(8.5,4.5)
\psline[linewidth=1pt]{->}(8.5,4.5)(10,3)
\psline[linewidth=1pt]{->}(8.5,4.5)(10,6)
\psline[linewidth=1pt]{->}(11,2.5)(14.5,2.5)
\psline[linewidth=1pt]{->}(11,6.5)(14.5,6.5)
\psline[linewidth=1pt](2.5,4.5)(2.5,1.5)
\psline[linewidth=1pt]{->}(2.5,1.5)(8,1.5)
\psline[linewidth=1pt](9,1.5)(10.5,1.5)
\psline[linewidth=1pt]{->}(10.5,1.5)(10.5,2)
\psline[linewidth=1pt](8.5,2)(8.5,3.5)
\psline[linewidth=1pt](8.5,3.5)(9.35,3.5)
\psarc[linecolor=black,linewidth=1pt](9.5,3.5){0.0937}{0}{180}
\psline[linewidth=1pt](9.65,3.5)(10.5,3.5)
\psline[linewidth=1pt]{->}(10.5,3.5)(10.5,6)
\rput[u](-1.5,2.5){$\#2$}
\rput[u](-1.5,6.5){$\#1$}
\rput[u](10.5,2.5){$\#2$}
\rput[u](10.5,6.5){$\#1$}
\rput[u](4.5,4.5){$\C$}
\rput[u](8.5,1.5){$\daleth$}
\rput[u](7.5,7){$\textcolor[rgb]{0.0,0.00,1.00}{|\mc{S}_0|=3}$}
\rput[u](2,7){$\textcolor[rgb]{0.0,0.00,1.00}{|\mc{A}|=4}$}
\rput[u](-0.5,2.5){$\bd{a}_2^n$}
\rput[u](-0.5,6.5){$\bd{a}_1^n$}
\rput[u](2,5.2){$\bd{a}^n=(\bd{a}_1^n,\bd{a}^n_2)$}
\rput[u](13,7){$\hat{\bd{a}}^n=(\hat{\bd{a}}_1^n,\hat{\bd{a}}^n_2)$}
\rput[u](13,3){$\hat{\bd{a}}^n=(\hat{\bd{a}}_1^n,\hat{\bd{a}}^n_2)$}
\rput[l](11,1.6){$\bd{s}_2^n,\bd{a}_2^n$}
\rput[l](11,5.6){$\bd{s}_1^n,\bd{a}_1^n$}
\rput[u](6,5){$\bd{s}_0^n$}
\end{pspicture}
\caption{This figure illustrates the encoder-assisted monitoring structure
for the repeated version of the prisoner's dilemma. The Theorem \ref{theo:FolkEncodeurAssiste}
provide a set $\RR$ of mixed strategies $\PP^{\star}\in\Delta(\mc{A})$ that allow the encoder $\C$
to establish VPM and the players $\mc{K}$ to implement an equilibrium strategy.
}\label{fig:ReconstructionRobustSideCapa2}
\end{center}
\end{figure}
To have a better understanding on how the results derived
in Sec. \ref{sec:info-constraints} and \ref{sec:game-model}
are exploited here, we consider a particular monitoring structure
$\daleth$ described by Fig. \ref{fig:MonitoringDilemme} with $\delta \in [0,1]$.
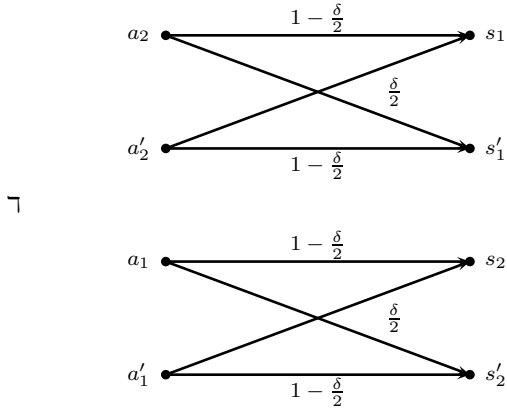
\begin{figure}[!ht]
\begin{center}
\psset{xunit=1cm,yunit=0.75cm}
\begin{pspicture}(0,-0.5)(4,6.5)
\psdot(0,0)
\psdot(0,2)
\psdot(4,0)
\psdot(4,2)
\psdot(0,4)
\psdot(0,6)
\psdot(4,4)
\psdot(4,6)
\psline[linewidth=1pt]{->}(0,0)(4,0)
\psline[linewidth=1pt]{->}(0,2)(4,2)
\psline[linewidth=1pt]{->}(0,4)(4,4)
\psline[linewidth=1pt]{->}(0,6)(4,6)
\psline[linewidth=1pt]{->}(0,0)(4,2)
\psline[linewidth=1pt]{->}(0,2)(4,0)
\psline[linewidth=1pt]{->}(0,4)(4,6)
\psline[linewidth=1pt]{->}(0,6)(4,4)
\rput[r](-0.2,0){$a_1'$}
\rput[r](-0.2,2){$a_1$}
\rput[r](-0.2,4){$a_2'$}
\rput[r](-0.2,6){$a_2$}
\rput[l](4.2,0){$s_2'$}
\rput[l](4.2,2){$s_2$}
\rput[l](4.2,4){$s_1'$}
\rput[l](4.2,6){$s_1$}
\rput[u](2,2.3){$1-\frac{\delta}{2}$}
\rput[u](2,6.3){$1-\frac{\delta}{2}$}
\rput[d](2,-0.3){$1-\frac{\delta}{2}$}
\rput[d](2,3.7){$1-\frac{\delta}{2}$}
\rput[u](3,1){$\frac{\delta}{2}$}
\rput[u](3,5){$\frac{\delta}{2}$}
\rput[u](-2,3){$\daleth$}
\end{pspicture}
\caption{Private monitoring structure $\daleth$ that depends on the parameter
$\delta \in [0,1]$.}\label{fig:MonitoringDilemme}
\end{center}
\end{figure}
This means that if the action $a_{-k} \in \mc{A}_{-k}$ was played, player $k \in\{1,2\}$ observes the right signal $s_k \in \mc{S}_k$ with probability $1-\frac{\delta}{2}$ and observes the wrong signal $s_k' \in \mc{S}_k$ with probability $\frac{\delta}{2}$. When $\delta = 0$, all the players have perfect monitoring. On the other hand, when $\delta = 1$, they cannot distinguish anything from the signal they observe (trivial monitoring). For this monitoring structure $\daleth(\bd{s}_1, \bd{s}_2|\bd{a}_1,\bd{a}_2)$ with $\delta \in [0,1]$, we want to determine the set $\conv u(\RR) \cap IR$ of utility profiles which are compatible with the information constraint (\ref{eq:RobustSide}). For the scenario under investigation, the information constraint (\ref{eq:RobustSide}) for $\delta>0$ rewrites as~:\\
\begin{footnotesize}
\begin{equation}
\begin{array}{cccl}
&\textsf{R}^{\star} = \max_{i\in\mc{K}} \bigg[\max_{k\in\mc{K},\atop a_i\in \mc{A}_i}H(\bd{a}_{-i,k}|\bd{s}_k(a_i), \bd{a}_k)+\log_2 \chi_i \bigg]
 &<  \log_2|\mc{S}_0| \\&&\nonumber\\
\stackrel{(a)}{\Leftrightarrow} & \ds{\max}
\left[
\begin{array}{c}
H(\bd{a}_{2}|\bd{s}_1) + \log_2\chi_1,\\
H(\bd{a}_{1}|\bd{s}_2) + \log_2\chi_2,
\end{array}\right]
&<  \log_2|\mc{S}_0|\\&&\nonumber\\
\stackrel{(b)}{\Leftrightarrow} &
\ds{\max}
\left[
\begin{array}{c}
 \sum_{a_{2} \in \mc{A}_{2},\atop  s_1\in \mc{S}_1} \PP^{\star}(a_{2})\daleth(s_1|a_{2})
\log_2 \frac{\sum_{\tilde{a}_{2} \in \mc{A}_{2}}\PP^{\star}(\tilde{a}_{2})\daleth(s_1|\tilde{a}_2)}{\PP^{\star}(a_{2})\daleth(s_1|a_2)} + \log_2\chi_1,\\
 \sum_{a_{1} \in \mc{A}_{1},\atop  s_2\in \mc{S}_2} \PP^{\star}(a_{1})\daleth(s_2|a_{1})
\log_2 \frac{\sum_{\tilde{a}_{1} \in \mc{A}_{1}}\PP^{\star}(\tilde{a}_{1})\daleth(s_2|\tilde{a}_1)}{\PP^{\star}(a_{1})\daleth(s_2|a_1)} + \log_2\chi_2
\end{array}\right]
&<  \log_2|\mc{S}_0| \\&&\nonumber\\
\stackrel{(c)}{\Leftrightarrow} & \ds{\max}
\left[
\begin{array}{cc}
 & \PP^{\star}(a_{2})(1 - \frac{\delta}{2}) \cdot \log_2\bigg(\frac{\PP^{\star}(a_{2})(1 - \frac{\delta}{2}) + \PP^{\star}(a_{2}') \frac{\delta}{2}}{\PP^{\star}(a_{2})(1 - \frac{\delta}{2})}\bigg)\\
+& \PP^{\star}(a_{2}') \frac{\delta}{2} \cdot \log_2\bigg(\frac{\PP^{\star}(a_{2})(1 - \frac{\delta}{2}) + \PP^{\star}(a_{2}') \frac{\delta}{2}}{\PP^{\star}(a_{2}') \frac{\delta}{2} }\bigg)\\
+& \PP^{\star}(a_{2})\frac{\delta}{2} \cdot \log_2\bigg(\frac{ \PP^{\star}(a_{2})\frac{\delta}{2} + \PP^{\star}(a_{2}')(1 - \frac{\delta}{2}) }{\PP^{\star}(a_{2})\frac{\delta}{2} }\bigg)\\
+& \PP^{\star}(a_{2}')(1 - \frac{\delta}{2}) \cdot \log_2\bigg(\frac{ \PP^{\star}(a_{2})\frac{\delta}{2} + \PP^{\star}(a_{2}')(1 - \frac{\delta}{2}) }{\PP^{\star}(a_{2}')(1 - \frac{\delta}{2}) }\bigg),\\
 & \PP^{\star}(a_{1})(1 - \frac{\delta}{2}) \cdot \log_2\bigg(\frac{\PP^{\star}(a_{1})(1 - \frac{\delta}{2}) + \PP^{\star}(a_{1}') \frac{\delta}{2}}{\PP^{\star}(a_{1})(1 - \frac{\delta}{2})}\bigg)\\
+& \PP^{\star}(a_{1}') \frac{\delta}{2} \cdot \log_2\bigg(\frac{\PP^{\star}(a_{1})(1 - \frac{\delta}{2}) + \PP^{\star}(a_{1}') \frac{\delta}{2}}{\PP^{\star}(a_{1}') \frac{\delta}{2} }\bigg)\\
+& \PP^{\star}(a_{1})\frac{\delta}{2} \cdot \log_2\bigg(\frac{ \PP^{\star}(a_{1})\frac{\delta}{2} + \PP^{\star}(a_{1}')(1 - \frac{\delta}{2}) }{\PP^{\star}(a_{1})\frac{\delta}{2} }\bigg)\\
+& \PP^{\star}(a_{1}')(1 - \frac{\delta}{2}) \cdot \log_2\bigg(\frac{ \PP^{\star}(a_{1})\frac{\delta}{2} + \PP^{\star}(a_{1}')(1 - \frac{\delta}{2}) }{\PP^{\star}(a_{1}')(1 - \frac{\delta}{2}) }\bigg),\\
\end{array}\right]
&<  \log_2|\mc{S}_0| -1\\&&\nonumber\\
\end{array}
\end{equation}
\end{footnotesize}
where (a) follows from the fact that $\bd{a}_{2}$ and $\bd{s}_1$ are independent of $\bd{a}_{1}$ and then the entropy
$H(\bd{a}_{2}|\bd{s}_1, \bd{a}_{1})$ reduce to $H(\bd{a}_{2}|\bd{s}_1)$. Using the same argument, $H(\bd{a}_{1}|\bd{s}_2, \bd{a}_{2})$ reduce to $H(\bd{a}_{1}|\bd{s}_2)$. (b) follow from the definition of the conditional entropy  and (c) follow the fact that the chromatic number of the graphs $\mc{G}_1$, $\mc{G}_2$ of both players are equals to $\chi_1 = \chi_2=2$ as soon as $\delta>0$.
\begin{figure}[h!l]
\vspace{1.7cm}
\includegraphics[width=0.85\textwidth]{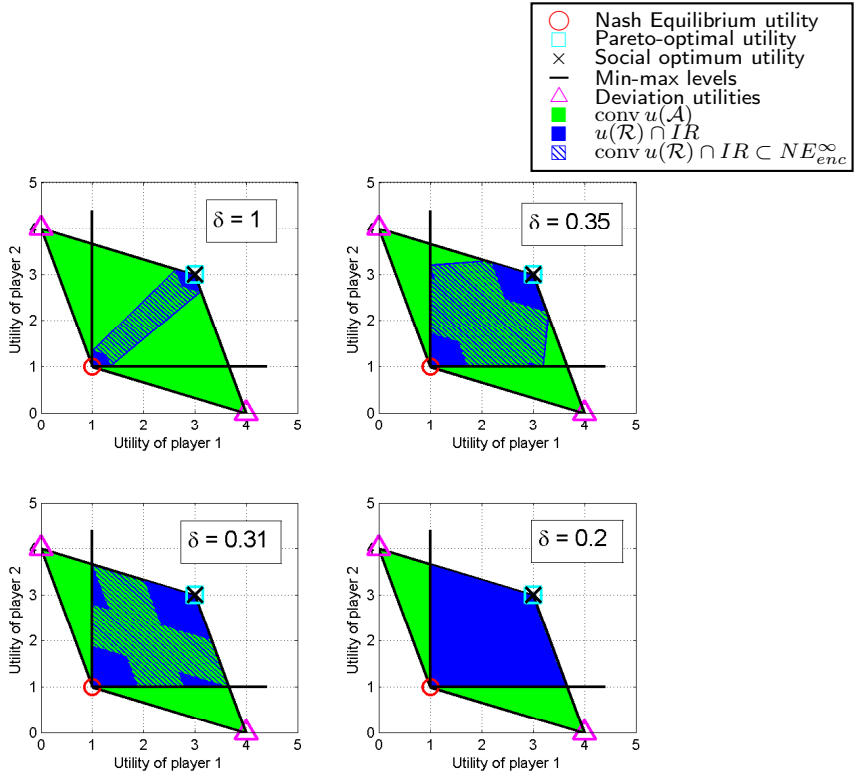}
\caption{The repeated version of the prisoner's dilemma is considered. The encoder-assisted monitoring structure of the game is described in
Fig. \ref{fig:ReconstructionRobustSideCapa2} where the private monitoring $\daleth$ of the players is described by Fig. \ref{fig:MonitoringDilemme}  and depends on a parameter $\delta\in[0,1]$. For $\delta=\{0.2,0.31,0.35,1\}$, the blue region represents the set $u(\RR) \cap IR$ of utilities that satisfy the information constraint (\ref{eq:RobustSide}). The hatched blue region represents the convexe hull $\conv u(\RR) \cap IR$ of the utility that can be supported by a uniform equilibrium strategy (Theorem \ref{theo:FolkEncodeurAssiste}). In that case, the encoder can maintain virtually perfect monitoring even if one of the players deviates. Note that for $\delta\leq 0.31$, the precision of the private monitoring is sufficient to guarantee the same equilibrium utility region as for the Folk theorem with perfect monitoring.}\label{fig:RobustPayoffRegion4}
\psset{xunit=0.5cm,yunit=0.25cm}
\begin{pspicture}(-15.5,-53)(-9.5,-48)
\psframe[fillcolor=white](0,0)(8.5,9)
\rput(0.75,8){$\textcolor[rgb]{0.98,0.00,0.00}{\bigcirc}$}
\rput(0.75,7){$\textcolor[rgb]{0.00,1.00,1.00}{\Box}$}
\rput(0.75,6){$\times$}
\rput(0.75,4){$\textcolor[rgb]{1.00,0.00,1.00}{\bigtriangleup}$}
\rput(0.75,3){$\textcolor[rgb]{0.00,1.00,0.00}{\blacksquare}$}
\rput(0.75,2){$\textcolor[rgb]{0.00,0.00,1.00}{\blacksquare}$}
\rput(0.75,1){$\textcolor[rgb]{0.00,0.00,1.00}{\square}$}
\psline(0.5,5)(1,5)
\psline[linewidth=0.5pt,linecolor=blue](0.925,0.65)(0.575,1.35)
\psline[linewidth=0.5pt,linecolor=blue](0.75,0.65)(0.575,1)
\psline[linewidth=0.5pt,linecolor=blue](0.925,1)(0.75,1.35)
\rput[l](1.5,8){ \sf{Nash Equilibrium utility}}
\rput[l](1.5,7){ \sf{Pareto-optimal utility}}
\rput[l](1.5,6){ \sf{Social optimum utility}}
\rput[l](1.5,5){ \sf{Min-max levels}}
\rput[l](1.5,4){ \sf{Deviation utilities}}
\rput[l](1.5,3){ \sf{$\conv u(\mc{A})$}}
\rput[l](1.5,2){ \sf{$u(\RR) \cap IR$}}
\rput[l](1.5,1){ \sf{$\conv u(\RR) \cap IR \subset NE_{enc}^{\infty}$}}
\end{pspicture}
\end{figure}

Setting $\delta$  to $1$, $0.35$, $0.31$ and $0.2$, the above information constraint can be translated into Fig. \ref{fig:RobustPayoffRegion4}. This figure represents the set of feasible average utility profiles which are both individually rational and compatible with the information constraint
(\ref{eq:RobustSide}). Let us interpret these numerical results that depend on the precision parameter $\delta\in[0,1]$ of the private monitoring $\daleth$.
\begin{itemize}
\item[$\circ$] \emph{Trivial monitoring:} $\delta=1$. The players have no information from their private signal, about the actions of their opponent. Theorem \ref{theo:FolkEncodeurAssiste} show that for some utility vectors represented by the blue hatched region $\conv u(\RR) \cap IR$, the encoder is able to send to both players, the sequences of past actions (with $|\mc{A}| = 4$) using an alphabet of $3 = |\mc{S}_0|$ symbols of public signals.
\item[$\circ$] \emph{Noisy imperfect monitoring:} $0.31<\delta<1$. The private signals received by the players reveal a partial information about the past actions of the opponent. Only a portion of the utility region is compatible with the information constraint (\ref{eq:RobustSide}). The virtual perfect monitoring and the equilibrium condition are not always implementable.
\item[$\circ$] \emph{Less noisy imperfect monitoring:} $0<\delta\leq0.31$. The blue hatched utility region $\conv u(\RR) \cap IR$ is equal to the utility region of the Folk theorem \cite{aumann-1981} with perfect monitoring $\conv u(\mc{A}) \cap IR$.
\item[$\circ$] \emph{Perfect monitoring:} $\delta=0$. The utility region coincides with the set of feasible and individually rational utilities.
\end{itemize}
The proposed approach provides an equilibrium strategy $(\sigma^{\star},\tau^{\star}) \in \Sigma \times\T$ that supports any utility profile in the blue hatched region $U^{\star} = \conv u(\RR) \cap IR$ of Fig. \ref{fig:RobustPayoffRegion4}, while ensuring that the encoder can maintain VPM even in the presence of single deviations.

\section{Conclusion}
\label{sec:concl}

This paper considers games where players have both a private signals they receive through the
initial monitoring structure and a public signal which is sent by an encoder. The encoder
is assumed to perfectly monitors the played actions and to send a public signal to the players.
The purpose of the encoder is to establish virtual perfect monitoring.
Technically, the encoder to be designed takes into account the side information at the receiver and possesses
the property of resilience to single deviations. It is shown that, the internal information
constraint imposes a restriction in terms of feasible utilities in order to establish virtual perfect monitoring and
provide an equilibrium utility region (as proved in the case of infinitely repeated games).

The proposed work can be extended in many respects. The targeted monitoring structure can be chosen to be different (e.g., a $2-$connected observation graph or a given public signal). The proposed information constraint might be relaxed by assuming that the encoder sends complementary private signals. An interesting result would be to establish a converse, proving that the information constraint is necessary and sufficient. The i.i.d assumption might be relaxed with the aim to characterize equilibrium utilities which do not assume i.i.d action profiles.


\appendix

\section{Proof of Theorem \ref{theo:RobustSide}}\label{sec:DemoTheoRobustSide}

We construct a coding scheme based on  graph coloring and statistical tests. Two points have to be considered carefully.
First, the side information  $s_k^n \in \mc{S}_k^n$ may provide some relevant information for player $k$ even if another player $i\in \mc{K}$ deviates.
Second, the transition probability $\daleth$ that generates the side information $s_k \in \mc{S}_k$ is controlled by the actions $a_k\in \mc{A}_k$
of each player $k\in\mc{K}$.\\\\
\bd{Parameter.} We choose a parameter $\varepsilon>0$ such that:
\begin{equation}
\textsf{R}^{\star} + 2 \varepsilon = \max_{i\in\mc{K}} \bigg[\max_{k\in\mc{K},\atop a_i\in \mc{A}_i}H(\bd{a}_{-i,k}|\bd{s}_k(a_i), \bd{a}_k)+\log_2 \chi_i \bigg]+ 2 \varepsilon \leq \log_2|\mc{S}_0|.
\end{equation}
\bd{Encoding function $f_0$.} The encoder proceeds to the  statistical test provided by (\ref{equation:StatisticalTest}) and constructs for a given sequence of actions $a^n=(a_1^n,\ldots,a_K^n) \in \mc{A}^n$, the following set~:
\begin{eqnarray}\label{equation:StatisticalTest}
\arg\min_{k\in \mc{K}}\sum_{a_{-k}\in \mc{A}_{-k}}\bigg|\frac{N(a_{-k}|a_{-k}^n)}{n}-\PP^{\star}_{-k}(a_{-k})\bigg|.
\end{eqnarray}
It chooses one component $i\in \mc{K}$ that minimizes (\ref{equation:StatisticalTest}).
The symbols of the component  $i\in \mc{K}$ will be encoded using the minimal coloring
$\phi_i: \mc{A}_i \longrightarrow \Phi_i$ (Def. \ref{def:Coloring})
of the  graph $\G_i$ defined by Def. (\ref{def:graph}). Denote by $\chi_i$ the chromatic number
of the graph $\G_i$ and
\begin{itemize}
\item[$\bullet$] encode the index of the chosen component $i \in \mc{K} $ using $|K|$ sequences $s_0^n \in \mc{S}_0^n$ of public signals~;
\item[$\bullet$] encode the sequence of colors $c_i^n \in \Phi_i^n$ that corresponds to the sequence of actions $a_i^n\in \mc{A}_i^n$ with at each stage $c_i=\phi_i(a_i)$ using $\chi_i^n$ sequences $s_0^n \in \mc{S}_0^n$ of public signals.
\end{itemize}
The other components $a_{-i}^n \in \mc{A}_{-i}^n$ will be encoded depending on the transition probability $\daleth$ and on the sequence $a_i^n \in \mc{A}_i^n$.
For example, if the symbol $a_i \in \mc{A}_i^n$ has been used at a high enough frequency, the sequences of signals $(s_k^n(a_i))_{k\in\mc{K}}$,
drawn from the transition probability $\daleth_{a_i} : \mc{A}_{-i} \longrightarrow \Delta(\mc{S}_k)$, are
sufficiently long to use a source coding scheme of the type Slepian and Wolf \cite{slepian-it-1973}.
Otherwise, the information $a_{-i}^n \in \mc{A}_{-i}^n$ should be encoded directly, without any compression.
The encoder splits the sequences $s_k^n \in \mc{S}_k^n$ into sub-sequences $(s_k^{n_{a_i}})_{a_i\in \mc{A}_i}$
indexed by the symbols $a_i\in \mc{A}_i$ where $n_{a_i} = N(a_i|a_i^n)$. The sub-sequence $s_k^{n_{a_i}} \in \mc{S}_k^{n_{a_i}}$
has length $n_{a_i}\in \N$ and is drawn i.i.d. from the joint probability
$\PP^{\star}_{-i} \otimes \daleth_{a_i} \in \Delta(\mc{A}_{-i} \times \mc{S})$. The encoder evaluates
the partition $(\tilde{\mc{A}}_i,\tilde{\mc{A}}_i^c)$ of the symbols $a_i\in \mc{A}_i$ defined as follows.
For each $\varepsilon>0$, there exists an $\bar{n}_1$ such that the error probability of the Slepian and
Wolf \cite{slepian-it-1973} coding is upper bounded by $\varepsilon > 0$~:
\begin{itemize}
\item[$\bullet$] $a_i\in \tilde{\mc{A}}_i$, if $N(a_i|a_i^n)=n_{a_i}\leq \bar{n_1}$ and then the sequence $a_{-i}^{n_{a_i}} \in \mc{A}_{-i}^{n_{a_i}}$ is encoded with $|\mc{A}_{-i}|^{n_{a_i}}$ sequences $s_0^n \in \mc{S}_0^n$ of public signals.
\item[$\bullet$] $a_i\in \tilde{\mc{A}}_i^c$ if $N(a_i|a_i^n)=n_{a_i}> \bar{n_1}$ and then the sequence $a_{-i}^{n_{a_i}} \in \mc{A}_{-i}^{n_{a_i}}$ is encoded using the "random binning technique" of Slepian and Wolf  \cite{slepian-it-1973}. \\
\end{itemize}
\bd{The random binning technique} \cite{slepian-it-1973} consists in randomly assign the  $2^{n_{a_i}H(\bd{a}_{-i})}$ typical
sequences $a_{-i}^{n_{a_i}}\in\mc{A}_{-i}^{n_{a_i}}$ to one of the
$2^{n_{a_i}(\max_{k\in\mc{K}} H(\bd{a}_{-i,k}|\bd{s}_k(a_i), \bd{a}_k)+ \varepsilon)}$ bin. Note that
$H(\bd{a}_{-i,k}|\bd{s}_k(a_i), \bd{a}_k) = H(\bd{a}_{-i}|\bd{s}_k(a_i), \bd{a}_k)$.
Each bin $\mc{B}(s_0^n)$ is indexed by a sequence $s_0^n \in \mc{S}_0^n$ of public signals and
contains $2^{n_{a_i}(\min_{k \in \mc{K}}I(\bd{a}_{-i};\bd{s}_k(a_i),\bd{a}_k)-\varepsilon)}$ typical
sequences $a_{-i}^{n_{a_i}}\in\mc{A}_{-i}^{n_{a_i}}$. The encoder $\C$ observes a sequence of realized
actions $a_{-i}^{n_{a_i}}\in\mc{A}_{-i}^{n_{a_i}}$. If this sequence is typical, then it send  to all
the players $\mc{K}$, the sequence of public signals $s_0^n\in\mc{S}_0^n$ corresponding to the bin
containing the sequence $a_{-i}^{n_{a_i}} \in \mc{B}(s_0^n)$.
If it is not, the encoder $\C$ declares an error.\\\\
\bd{Decoding function $g_k$ of player $k\in\mc{K}$.} The decoding player receives the index
$i \in \mc{K} $ of the player chosen by the statistical test (\ref{equation:StatisticalTest}).
Using the appropriate codebook, it decodes separately the information regarding the component
$i \in \mc{K} $ and the other components $j\in \mc{K}\backslash\{i\}$.
\begin{itemize}
\item[$\bullet$] Knowing component $i\in \mc{K}$ chosen by the statistical test, the side information $s_k\in \mc{S}_k$ and the color $c_i\in \Phi_i$, the decoding player $k\in\mc{K}$ decodes a unique stage symbol $a_i\in \mc{A}_i$ for component $i \in \mc{K} $.
\end{itemize}
The decoding player $k\in \mc{K}$ knows the entire sequence of actions $a_i^n \in \mc{A}_i^n$ and it characterizes the partition  $\tilde{\mc{A}_i}$ and $\tilde{\mc{A}_i}^c$ of the set of symbols $\mc{A}_i$.
\begin{itemize}
\item[$\bullet$] For the transition $\daleth_{a_i}$, controlled by the symbol $a_i\in \tilde{\mc{A}_i}$, the sequence of symbols $a_{-i}^{n_{a_i}} \in \mc{A}_{-i}^{n_{a_i}}$ is directly decoded.
\item[$\bullet$] For the transition $\daleth_{a_i}$, controlled by the symbol $a_i\in \tilde{\mc{A}_i}^c$, the sequence of actions $a_{-i}^{n_{a_i}}  \in \mc{A}_{-i}^{n_{a_i}}$ is decoded using Slepian and Wolf decoding \cite{slepian-it-1973}. The decoding player $k\in\mc{K}$ find into the bin $\mc{B}(s_0^n)$ corresponding to the sequence of public signals $s_0^n\in\mc{S}_0^n$, a sequence $a_{-i}^{n_{a_i}}  \in \mc{A}_{-i}^{n_{a_i}}$ which is jointly typical with the sequence of side information $s_k^{n}(a_i) \in \mc{S}_k^n$ for the probability distribution $\PP_{-i}\otimes \daleth_{a_i}\in \Delta(\mc{A}_{-i}\times \mc{S}_k)$.
\end{itemize}
\bd{Cardinality of $\mc{S}_0^n$.}
Let $\bar{n_2}>\frac{\log|K|+\bar{n_1}|\mc{A}_i|\log|A_{-i}|}{\varepsilon}$. Then for all $n\geq\bar{n_2}$, the cardinality of the set of sequences $|\mc{S}_0|^n$ is greater than the number of sequences of the coding scheme $2^{n(R^{\star}+3\varepsilon)}$.
\begin{eqnarray}
&&  \frac{\log \bigg(|\mc{K}|\cdot\chi_i^n\cdot |\mc{A}_{-i}|^{\sum_{a_i\in \tilde{\mc{A}_i}}n_{a_i}}\cdot\prod_{a_i\in \tilde{\mc{A}_i}^c}2^{n_{a_i}(\max_{k\in\mc{K}} H(\bd{a}_{-i,k}|\bd{s}_k(a_i), \bd{a}_k)+ \varepsilon)}\bigg)}{n}\nonumber \\
 &\leq &  \frac{\log|K|}{n} + \log\chi_i+ \frac{\bar{n_1}|\tilde{\mc{A}_i}|}{n}\log|\mc{A}_{-i}| + \sum_{a_i\in \tilde{\mc{A}_i}^c}\frac{n_{a_i}}{n}(\max_{k\in\mc{K}} H(\bd{a}_{-i,k}|\bd{s}_k(a_i), \bd{a}_k)+ \varepsilon)\nonumber\\
&\leq&  \max_{i\in \mc{K}} \bigg[\max_{k\in\mc{K},\atop a_i\in\mc{A}_i} H(\bd{a}_{-i,k}|\bd{s}_k(a_i), \bd{a}_k) + \log\chi_i\bigg] + \frac{\log|K|+\bar{n_1}|\mc{A}_i|\log|A_{-i}|}{n}+\varepsilon\nonumber\\
 &\leq& \max_{i\in \mc{K}} \bigg[\max_{k\in\mc{K},\atop a_i\in\mc{A}_i} H(\bd{a}_{-i,k}|\bd{s}_k(a_i), \bd{a}_k) + \log\chi_i\bigg] + 2\varepsilon\nonumber \\
&=& \sf{R}^{\star}+3\varepsilon\\
&\leq& \log_2|\mc{S}_0|.
\end{eqnarray}
\bd{Error probability.}
Suppose that player $i\in \mc{K}$ chooses his sequence of actions $a_i^n \in \mc{A}_i^n$ with an arbitrary sequence of distribution. There are two possibilities. First, the statistical test (\ref{equation:StatisticalTest}) returns the deviating player $i\in \mc{K}$.
Second the statistical test returns another player $j\neq i$.
\begin{itemize}
\item[$\bullet$]
Suppose that the statistical test (\ref{equation:StatisticalTest}) returns the deviating player $i\in \mc{K}$.
In that case, the "random binning technique" of Slepian and Wolf \cite{slepian-it-1973} guarantees that for all $a_i \in \tilde{\mc{A}}_i^c$ the sequence of vectors of actions
$a_{-i}^{n_{a_i}} \in \mc{A}_{-i}^{n_{a_i}}$ is perfectly reconstructed with large probability. Let us define the following events:
\begin{eqnarray}
E_1&=&  \cup_{k\in\mc{K},\atop a_i \in \mc{A}_i} \bigg\{ (\bd{a}_{-i}^{n_{a_i}}, \bd{s}_k^{n_{a_i}})\notin A_{\varepsilon}^{{\star}{n}}( \mc{A}_{-i}\times \mc{S}_k)\bigg\}.
\end{eqnarray}
There exists a player $k\in\mc{K}$ for which the random sequences of actions and private signals $(\bd{a}_{-i}^{n_{a_i}}, \bd{s}_k^{n_{a_i}}) \in \mc{A}_{-i}^{n_{a_i}} \times \mc{S}_k^{n_{a_i}}$ are not typical.
\begin{eqnarray}
E_2&=&  \cup_{k\in\mc{K},\atop a_i \in \mc{A}_i} \bigg\{ \exists \bd{a}_{-i}^{'n_{a_i}}\neq \bd{a}_{-i}^{n_{a_i}} \in \mc{B}(s_0^{n_{a_i}}) ,\; (\bd{a}_{-i}^{'n_{a_i}},\bd{s}_k^{n_{a_i}}(a_i), \bd{a}_k^{n_{a_i}})\in A_{\varepsilon}^{{\star}{n}}(\mc{A}_{-i}\times \mc{S}_k)\bigg\}.\nonumber \\
\end{eqnarray}
There exists another sequence $\bd{a}_{-i}^{'n_{a_i}}$ in the bin $\mc{B}(s_0^{n_{a_i}})$ corresponding to the sequence of public signals $s_0^{n_{a_i}}\in\mc{S}_0^{n_{a_i}}$ that is jointly typical with the sequences of private signals $\bd{s}_k^{n_{a_i}}(a_i)$ and actions $\bd{a}_k^{n_{a_i}}$ of the player $k\in\mc{K}$.\\
\begin{itemize}
\item[$\circ$]From Lemma \ref{lemma:TypicalSequencesCond}  of App. \ref{sec:AppTypicaSeq}, the error probability $\PP(E_1)$ is lower than  $\varepsilon \cdot K\cdot |\mc{A}_i| >0$ as soon as $n$ is sufficiently large.
\item[$\circ$]From Lemma \ref{lemma:MutualProb} of App. \ref{sec:AppTypicaSeq}, the expected error probability $\E_{\lambda}[\PP(E_2)]$ of the random code $\mu \in \Delta(\Lambda(n))$ is lower than  $\varepsilon \cdot K\cdot |\mc{A}_i| >0$ as soon as $n$ is sufficiently large and the condition (\ref{eq:condMutualCovering}) is satisfied.
\begin{eqnarray}
 |\mc{B}(s_0^{n_{a_i}})| \leq 2^{n (\min_{k \in \mc{K}}I(\bd{a}_{-i};\bd{s}_k(a_i),\bd{a}_k)-\varepsilon)} \label{eq:condMutualCovering}
\end{eqnarray}
Lemma \ref{lemma:MutualProb} applies because the random sequence
$\bd{a}_{-i}^{'n_{a_i}}$ is generated independently
of the random sequences $(\bd{s}_k^{n_{a_i}}(a_i), \bd{a}_k^{n_{a_i}})$.
This ensures the existence of a code $\lambda \in \Lambda(n)$ such that the
error probability  $\PP_{\lambda}(E_2) \leq 2 \varepsilon$ is upper bounded.\\
\end{itemize}
\item[$\bullet$]
Suppose that the the statistical test returns another player $j\neq i \in \mc{K} $.
This implies the following inequality:\\
\begin{eqnarray}\label{equation:StatisticalTestJoueurs}
&&\sum_{a_{-j}\in \mc{A}_{-j}}\bigg|\frac{N(a_{-j}|a_{-j}^n)}{n}-\PP_{-j}(a_{-j})\bigg| \leq \sum_{a_{-i}\in \mc{A}_{-i}}\bigg|\frac{N(a_{-i}|a_{-i}^n)}{n}-\PP^{\star}_{-i}(a_{-i})\bigg|.\nonumber\\ \\ \nonumber
\end{eqnarray}
For every player $j\in \mc{K}\backslash \{i\}$, the sequence $a_j^n \in \mc{A}_j^n$ is drawn i.i.d. from stage to stage with the distribution
$\PP_j^{\star} \in \Delta(\mc{A}_j)$. From Lemma \ref{lemma:TypicalSequencesCond},
these action sequences are typical with large probability as $n$ goes to infinity.
Then, the sequence of actions $a_i^n \in \mc{A}_i^n$ is typical with large probability and then correctly encoded and decoded. There exists $n$ sufficiently large such that the
error probability $\PP_e(\lambda) \leq \varepsilon $ is upper bounded.
\end{itemize}
We therefore proved the existence a code $\lambda \in \Lambda(n)$ such that the error probability of the code $\PP_e(\lambda) \leq 2 \varepsilon \cdot K\cdot |\mc{A}_i| $ is upper bounded.

\section{Proof of Theorem \ref{theo:FolkEncodeurAssiste}}\label{sec:DemoJeu}

We prove the following inclusion $\conv u(\RR) \cap IR \subset NE^{\infty}_{enc} $.
First, we consider a utility vector $U\in u(\RR) \cap IR$ and provide a
pair of strategies for the encoder and the players
$(\sigma^{\star}, \tau^{\star})\in \Sigma \times \T $
that forms a uniform equilibrium (see Def. \ref{def:EqInfini}.
The first condition (i) is satisfied when the asymptotic utility of the strategies
 $(\sigma^{\star}, \tau^{\star})\in \Sigma \times \T $ converges toward the utility $U$.
The second condition (ii) is satisfied when no unilateral deviation $\tau'_k\in \T_k $
provides to player $k\in\mc{K}$ a gain larger than $\epsilon>0$.

\subsection{Construction of strategies $(\sigma^{\star},\tau^{\star}) \in \Sigma \times \T $}\label{sec:DemoConstructionStrat}
\subsubsection{Block coding scheme}\label{sec:DemoSchemaCodage}

The $T>0$ stages of the repeated game are divided into $B$ blocks
of stages of length $n$, represented by Fig. \ref{fig:GameCourse}.
Denote $\mc{B}$ the set of blocks,
$b\in\mc{B}$ the index of one block and $B = |\mc{B}|\in \N$ the number of blocks.
Denote  $s_k^n(b)\in\mc{S}_k^{n}$ the sequence of signals received during the block $b\in \mc{B}$.
Fix the parameter $\epsilon>0$ and let us describe the strategies
 $(\sigma^{\star},\tau^{\star})\in\Sigma \times\T$ that satisfy both conditions
(\ref{eq:DemoCondition1}) and (\ref{eq:DemoCondition2}) for all $T \geq \bar{T}$.
\begin{eqnarray}
|\gamma_k^T(\sigma^{\star}, \tau^{\star}) - U_k^{\star} |\qquad \leq& \epsilon,&\qquad \forall k\in\mc{K},\label{eq:DemoCondition1}\\
\gamma_k^T(\sigma^{\star}, \tau^{\star}) + \epsilon \qquad\geq& \gamma_k^T(\sigma^{\star}, \tau_k', \tau^{\star}_{-k}),&\qquad \forall k\in\mc{K},\quad \forall \tau_k'\in\T_k.\label{eq:DemoCondition2}
\end{eqnarray}
Suppose that the number of blocks $B\in\N$ satisfies condition (\ref{eq:DemoConditionB})~:
\begin{eqnarray}
B\geq \frac{8 \cdot\max_{a\in\mc{A}} |u_k(a)|}{\epsilon}. \label{eq:DemoConditionB}
\end{eqnarray}

\subsubsection{Strategy of the encoder $\sigma^{\star} \in \Sigma $}\label{sec:DemoStrategieCodage}
The coding strategy $\sigma^{\star}\in \Sigma $ consists in sending
a sequence of public signals $\bd{s}_0^n \in \mc{S}_0^n$ to each player
 so that they can reconstruct the sequence $\bd{a}^n \in \mc{A}^n$ of past actions.
In order to communicate, the encoder $\C$ and the players $\mc{K}$ implement a code $\lambda = (f_0,(g_k)_{k\in\mc{K}})$
investigated in Sec. \ref{sec:info-constraints} and
defined by~:
\begin{equation}\label{eq:Encoding2}
\begin{array}{cccccc}
f_0 & : & \mc{A}^n   & \longrightarrow & \mc{S}_0^n&,\\
g_k & : & \mc{S}_0^n  \times \mc{S}_k^{n} \times \mc{A}_k^n & \longrightarrow &\mc{A}^n,&\qquad \forall k\in\mc{K}.
\end{array}
\end{equation}
Condition $U\in \conv u(\RR) \cap IR$ implies that the probability distribution
$\PP^{\star} \in \prod_{k\in\mc{K}} \Delta(\mc{A}_k)$ belong to the set
$\RR$ described by (\ref{eq:CondDefMixte}) and satisfies the condition (\ref{eq:RobustSide})
of the Theorem $\ref{theo:RobustSide}$. This coding result implies that for all
 $\varepsilon>0$, there exists a parameter $n\in\N$ and a code
$\lambda \in  \Lambda(n)$ with $\lambda = (f_0,(g_k)_{k\in\mc{K}})$ such that the error probability
of the coding scheme is upper bounded by $\varepsilon$.
Denote by $\hat{{a}}^n(k)\in\mc{A}^n$ the sequence of actions
obtained as output by the decoder $k\in\mc{K}$.
\begin{eqnarray}
&\PP_e(\lambda) =& \sum_{k\in\mc{K}} \max_{i \in \mc{K}}\max_{v_i \in  \Delta(\mc{A}_i^{\infty})}
\PP_{v_i}(\bd{a}^n \neq g_k(\bd{s}_k^n,\bd{s}_0^n,\bd{a}_k^n)).\label{eq:ErrorProbaKStrat}
\end{eqnarray}
The strategy of the encoder $\sigma^{\star}\in\Sigma$
is built as follows. At the beginning of the block $b\in \mc{B}$ with $b\geq 2$, the encoder $\C$
observes the sequence of actions $a^n(b-1)\in\mc{A}^n$ over the block $b-1\in\mc{B}$ and
choose the sequence of public signals $s_0^n(b)$
over block $b\in\mc{B}$ using the encoding function $f_0$
(\ref{eq:Encoding2}) provided by the code
$\lambda$ that satisfies the condition (\ref{eq:ErrorProbaKStrat}).
\begin{eqnarray}
s_0^n(b) = f_0\bigg( a^n(b-1) \bigg) \in \mc{S}_0^n.
\end{eqnarray}
Over the first block $b_1\in\mc{B}$, the encoder send an arbitrary sequence $s_0(b_1)\in\mc{S}_0^n$.

\subsubsection{Decoding scheme}\label{sec:DemoDecodage}
At the end of the block $b\in \mc{B}$ with $b\geq 3$, player $k$ implements the decoding function
 $g_k$ (\ref{eq:Encoding2}) provided by the code $\lambda$
that satisfies the condition (\ref{eq:ErrorProbaKStrat}).
The player $k\in\mc{K}$ recalls his own actions $a_k^n(b-1) \in \mc{A}_k^n$ and
observes the sequences of private signals
$s_k^n(b-1) \in \mc{S}_k^n$ and public signals $s_0^n(b) \in \mc{S}_0^n$
sent by the encoder $\C$. The player $k\in\mc{K}$ evaluates the sequence
 $\hat{\bd{a}}^n(k,b-1)$ of actions of block $b-1\in \mc{B}$ using the decoding function $g_k$.
\begin{eqnarray}
\hat{\bd{a}}^n(k,b-1) = g_k\bigg(s_k^n(b-1), s_0^n(b), a_k^n(b-1)\bigg) \in \mc{A}^n.
\end{eqnarray}
Condition (\ref{eq:ErrorProbaKStrat}) guarantees that at the beginning of block
$b+1\in \mc{B}$, each player  $k\in\mc{K}$ observes the sequence of actions $\hat{\bd{a}}(b-1) \in \mc{A}$
of the other players during the block $b-1\in\mc{B}$ with an error probability arbitrarily low.
Over the two first blocks  $b_1, b_2\in\mc{B}$, no decoding strategy is implemented.

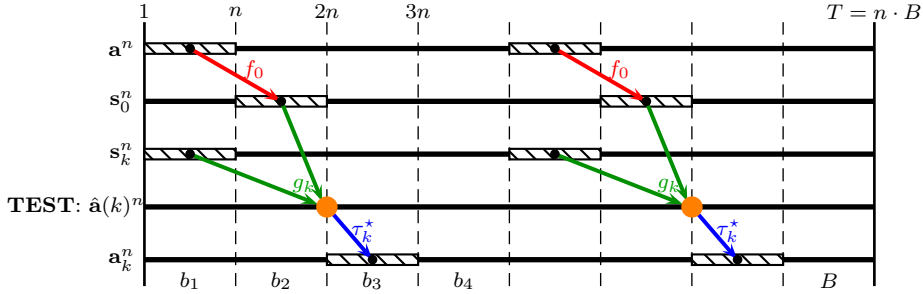
\begin{figure}[!ht]
\begin{center}
\psset{xunit=0.6cm,yunit=0.35cm}
\begin{pspicture}(-2,-1.5)(16,9.5)
\psline[linewidth=2pt](0,0)(16,0)
\psline[linewidth=2pt](0,2)(16,2)
\psline[linewidth=2pt](0,4)(16,4)
\psline[linewidth=2pt](0,6)(16,6)
\psline[linewidth=2pt](0,8)(16,8)
\psline[linewidth=1pt](0,-1)(0,9)
\psline[linewidth=0.5pt,linestyle=dashed ](2,-1)(2,9)
\psline[linewidth=0.5pt,linestyle=dashed ](4,-1)(4,9)
\psline[linewidth=0.5pt,linestyle=dashed ](6,-1)(6,9)
\psline[linewidth=0.5pt,linestyle=dashed ](8,-1)(8,9)
\psline[linewidth=0.5pt,linestyle=dashed ](10,-1)(10,9)
\psline[linewidth=0.5pt,linestyle=dashed ](12,-1)(12,9)
\psline[linewidth=0.5pt,linestyle=dashed ](14,-1)(14,9)
\psline[linewidth=1pt](16,-1)(16,9)
\rput[u](-0.5,0){$\bd{a}_{k}^n$}
\rput[u](-1.5,2){\textbf{TEST}: $\hat{\bd{a}}(k)^n$}
\rput[u](-0.5,4){$\bd{s}_k^n$}
\rput[u](-0.5,6){$\bd{s}_0^n$}
\rput[u](-0.5,8){$\bd{a}^n$}
\rput[u](6,9.4){$3n$}
\rput[u](4,9.4){$2n$}
\rput[u](2,9.4){$n$}
\rput[u](0,9.4){$1$}
\rput[u](16,9.4){$T =n \cdot B$}
\rput[u](1,-0.8){$b_1$}
\rput[u](3,-0.8){$b_2$}
\rput[u](5,-0.8){$b_3$}
\rput[u](7,-0.8){$b_4$}
\rput[u](15,-0.8){$B$}
\pspolygon[fillstyle=vlines*](0,7.8)(0,8.2)(2,8.2)(2,7.8)
\pspolygon[fillstyle=vlines*](0,3.8)(0,4.2)(2,4.2)(2,3.8)
\pspolygon[fillstyle=vlines*](2,5.8)(2,6.2)(4,6.2)(4,5.8)
\pspolygon[fillstyle=vlines*](4,-0.2)(4,0.2)(6,0.2)(6,-0.2)
\psline[linewidth=1.5pt,linecolor=red  ]{->}(1,8)(3,6)
\rput[u](2.4,7.2){\textcolor[rgb]{0.98,0.00,0.00}{$f_0$}}
\psline[linewidth=1.5pt,linecolor=dgreen ]{->}(3,6)(3.9,2.2)
\psline[linewidth=1.5pt,linecolor=dgreen ]{->}(1,4)(3.8,2.1)
\rput[u](3.5,2.8){\textcolor[rgb]{0.00,0.64,0.00}{$g_k$}}
\psline[linewidth=1.5pt,linecolor=blue ]{->}(4,2)(5,0)
\rput[u](4.8,1.1){\textcolor[rgb]{0.00,0.00,1.00}{$\tau^{\star}_k$}}
\psdots[linewidth=3pt,linecolor=orange](4,2)
\psdots(5,0)(1,4)(1,8)(3,6)
\pspolygon[fillstyle=vlines*](8,7.8)(8,8.2)(10,8.2)(10,7.8)
\pspolygon[fillstyle=vlines*](8,3.8)(8,4.2)(10,4.2)(10,3.8)
\pspolygon[fillstyle=vlines*](10,5.8)(10,6.2)(12,6.2)(12,5.8)
\pspolygon[fillstyle=vlines*](12,-0.2)(12,0.2)(14,0.2)(14,-0.2)
\psline[linewidth=1.5pt,linecolor=red  ]{->}(9,8)(11,6)
\rput[u](10.4,7.2){\textcolor[rgb]{0.98,0.00,0.00}{$f_0$}}
\psline[linewidth=1.5pt,linecolor=dgreen ]{->}(11,6)(11.9,2.2)
\psline[linewidth=1.5pt,linecolor=dgreen ]{->}(9,4)(11.8,2.1)
\rput[u](11.5,2.8){\textcolor[rgb]{0.00,0.64,0.00}{$g_k$}}
\psline[linewidth=1.5pt,linecolor=blue ]{->}(12,2)(13,0)
\rput[u](12.8,1.1){\textcolor[rgb]{0.00,0.00,1.00}{$\tau^{\star}_k$}}
\psdots[linewidth=3pt,linecolor=orange](12,2)
\psdots(13,0)(9,4)(9,8)(11,6)
\end{pspicture}
\caption{The strategies of the encoder $\C$ and the players $\mc{K}$
 $(\sigma^{\star},\tau^{\star})\in\Sigma \times \mc{T}$
are described at sections  \ref{sec:DemoStrategieCodage} and \ref{sec:DemoStratEquil}.
The actions $\bd{a}^n(b)$ over block $b\in\mc{B}$ are encoded over the next block
 $b+1 \in \mc{B}$ into a sequence of public signals $\bd{s}_0^n(b+1)$.
At the end of block $b+1 \in \mc{B}$, player $k\in\mc{K}$ decode the sequence of actions
$\hat{\bd{a}}^n(b)$ over block $b\in\mc{B}$ from the sequences of signals
$\bd{s}_0^n(b+1)$ and $\bd{s}_k^n(b)$. Player $k\in\mc{K}$ performs a statistical test
in order to detect the possible unilateral deviations.
The result of this statistical test determines the sequence of actions
 $\bd{a}_k^n(b+2)$ player $k\in\mc{K}$ will play during the block $b+2 \in\mc{B}$.}\label{fig:GameCourse}
\end{center}
\end{figure}

\subsubsection{Statistical test}\label{sec:DemoTestStat}
Each player $k\in \mc{K}$ performs a statistical test at the beginning of each block $b +1 \in \mc{B}$.
Define the event $\bd{E}_k^i(b+1)$ using the set of typical sequences
 $A_{\varepsilon}^{\star{n}}(\PP^{\star}_{i})$ stated by the definition \ref{def:TypicalSequencesCond}.
\begin{eqnarray}
\bd{E}_k^i(b+1) = \begin{cases}
0 \text{ if } \hat{a}_{i}^n(b-1) \in A_{\varepsilon}^{\star{n}}(\PP^{\star}_{i}),\\
1 \text{ if } \hat{a}_{i}^n(b-1) \notin A_{\varepsilon}^{\star{n}}(\PP^{\star}_{i}).
\end{cases}\label{eq:TestStat}
\end{eqnarray}
When $\bd{E}_k^i(b+1) = 1$, player $k\in \mc{K}$
declare player $i\in \mc{K}$ deviates from the prescribed strategy
$\tau^{\star}_i \in\mc{T}_i$, during block $b-1\in\mc{B}$.

\subsubsection{Main plan}\label{sec:DemoPlanPrincipal}
The main plan consists in playing the same mixed action $\PP^{\star}$ i.i.d. from stage to stage.
\begin{eqnarray}
\PP^{\star}(a^t)= \PP^{\star}(a_1)\otimes \ldots \otimes \PP^{\star}(a_K) \in \prod_{k\in\mc{K}} \Delta(\mc{A}_k),\qquad \forall t\geq 1.\label{PlanPrincipal}
\end{eqnarray}

\subsubsection{Punishment plan for player $i\in \mc{K}$}\label{sec:DemoPlanPunition}
The punishment plan $\bar{\PP}(i)=(\bar{\PP}_k(i))_{k \neq i}\in \prod_{k\neq i} \Delta(\mc{A}_k)$
of player $i\in \mc{K}$ consists of a vector of mixed actions of other players
that minimize the utility of player $i\in \mc{K}$.
\begin{eqnarray}
 \bar{\PP}(i) = \bigg(\bar{\PP}_k(i)\bigg)_{k\neq i} \in
 \argmin_{ \PP_{-i}\in \prod_{k\neq i} \Delta(\mc{A}_k)} \left[ \max_{\PP_i \in \Delta(\mc{A}_i)} \E_{\PP_i,\PP_{-i}}\bigg[u_k(a_i,a_{-i})\bigg]\right],\quad \forall i\in \mc{K}.\nonumber \\ \label{PlanPunition}
\end{eqnarray}
The punishment plan for player $i\in \mc{K}$ by player $k \in \mc{K}$ is denoted
$\bar{\PP}_k(i)\in \Delta(\mc{A}_k)$ and is given by (\ref{PlanPunition}).
If all the players $k\neq i$ play the strategy $\bar{\PP}(i) = (\bar{\PP}_k(i))_{k\neq i}$,
the player $i\in \mc{K}$ cannot obtain a utility greater than his min-max level $\upsilon_i\in \R$
characterized by (\ref{eq:NiveauMinmax}).

\subsubsection{Equilibrium strategy $\tau^{\star}=(\tau^{\star}_k)_{k\in \mc{K}} \in \T$}\label{sec:DemoStratEquil}

At the beginning of each block $b\geq 3 \in\mc{B}$, the equilibrium strategy is described
as follows:

\begin{itemize}
\item[$\bullet$] Player $k$ implements the decoding scheme
(Sec. \ref{sec:DemoDecodage}) and reconstructs the actions $\hat{a}_{-k}(b-2) \in \mc{A}_{-k}^n$
played by the other players $j\neq k$ during block $b-2 \in\mc{B}$.
\item[$\bullet$] Player $k$ implements the statistical test $\bd{E}_i^k$
defined section \ref{sec:DemoTestStat}, in order to detect possible unilateral deviations.
\item[$\bullet$] If the statistical test is negative, ($\forall b'\leq b,\;\forall i \in\mc{K},\quad\bd{E}_k^i(b) = 0$),
then player $k\in \mc{K}$ play the main plan $\PP^{\star}_k\in \Delta(\mc{A}_k)$ stated section
 \ref{sec:DemoPlanPrincipal} during every stage of block $b\in\mc{B}$.
\item[$\bullet$] If the statistical test is positive, ($\exists b'\leq b,\;\exists i \in\mc{K},\quad\bd{E}_k^i(b) = 1$),
then player $k\in \mc{K}$ play the punishment plan $\bar{\PP}_k(i)\in \Delta(\mc{A}_k)$ stated section
\ref{sec:DemoPlanPunition} corresponding to the player $i\in \mc{K}$ until the end of the last block $B\in\mc{B}$.
If several deviations are detected simultaneously $\bd{E}_k^i(b)=\bd{E}_j(b) = 1$,
then player $k\in \mc{K}$ punishes anyone of those players who is the
smaller, according to a total order over $\mc{K}$, previously fixed.
\end{itemize}
Over the first two blocks $b_1, b_2 \in\mc{B}$, players $\mc{K}$ play the main plan $\PP^{\star} \in \Delta(\mc{A})$.
The equilibrium strategy $\tau^{\star}=(\tau^{\star}_k)_{k\in \mc{K}} \in \T$
is defined at each stage $t\geq 1$ as follows:
\begin{eqnarray}
  \tau^{\star t}_k(h^t) &=& \left\{
          \begin{array}{lll}
           \PP^{\star}_k&\in \Delta(\mc{A}_k) \qquad &\text{while } \bd{E}_k^i(b) = 0, \quad  \forall i \neq k,\; \forall b \leq
           \lfloor\frac{t}{n} ,          \rfloor   \\
           \bar{\PP}_k(i)&\in \Delta(\mc{A}_k)  \qquad &\text{otherwise.}\\
          \end{array}
        \right.\label{StrategieOptimal}
\end{eqnarray}

\subsection{Condition (i) of definition \ref{def:EqInfini}: convergence of the utilities}\label{sec:DemoConvUtil}
Let us fix a parameter $\epsilon>0$ and prove that there exists a
$\bar{T}>1$ such that for all $T\geq \bar{T}$, the  utilities of the encoder and the players $\mc{K}$
 $(\sigma^{\star},\tau^{\star})\in\Sigma \times \mc{T}$
defined in Sec. \ref{sec:DemoStrategieCodage} and \ref{sec:DemoStratEquil},
are $\epsilon$-closed to the utility $U\in \conv u(\RR) \cap IR$.
Remark that $\epsilon>0$ and $\varepsilon>0$ are two distinct parameters.
Define the following event:
\begin{eqnarray}
\bd{E} = \begin{cases}
1 \text{ if } \qquad\qquad\exists b\in\mc{B}, \exists i,k\in\mc{K} \text{ such that }\bd{E}_k^i(b) = 1,\\
0 \text{ otherwise. } \label{eq:EvenementErreurJeu}
\end{cases}
\end{eqnarray}
When $\bd{E}=0$, then no unilateral deviation is detected during the course of the game.
\begin{lemma}\label{lemme:EvenementJeu}
Suppose that the encoder $\C$ and the players $\mc{K}$ implements the strategies
$(\sigma^{\star},\tau^{\star})\in\Sigma \times \mc{T}$.
Then for all $\varepsilon>0$, there exists a block length $n_1\in \N$, such that for all
$n\geq n_1$, the probability of event $\bd{E} = 1$ is bounded as follows:
\begin{eqnarray}
\PP(\bd{E} = 1) \leq 2\varepsilon \cdot B \cdot K^2. \label{eq:LemmeEvenementJeu}
\end{eqnarray}
\end{lemma}
The result of Lemma \ref{lemme:EvenementJeu} is useful for the proof of Lemma \ref{lemme:UtiliteEpsilon}.
\begin{lemma}\label{lemme:UtiliteEpsilon}
Suppose that the encoder $\C$ and the players $\mc{K}$ implement the strategies
$(\sigma^{\star},\tau^{\star})\in\Sigma \times \mc{T}$.
Then for all $\varepsilon>0$, there exists a block length $n_1\in \N$, such that for all
$n\geq n_1$, the expected utility satisfies the following equation~:
\begin{eqnarray}
\bigg|\gamma_k^T(\sigma^{\star},\tau^{\star}) - \E_{\PP^{\star}} \bigg[ u_k(\bd{a}_k,\bd{a}_{-k}) \bigg]\bigg|
\leq  4  \varepsilon \cdot \max_{a\in\mc{A}} |u_k(a)| \cdot  B\cdot K^2, \qquad \forall k\in\mc{K}. \label{eq:UtiliteEpsilon}
\end{eqnarray}
\end{lemma}
For the parameter $\epsilon>0$ and a fixed number of block $B\in\N$, there
exists a parameter $\varepsilon>0$ and a block length $n_1\in\N$ such that
$4  \varepsilon \cdot \max_{a\in\mc{A}} |u_k(a)| \cdot  B\cdot K^2 \leq \epsilon$.
From Lemma \ref{lemme:UtiliteEpsilon}, the strategy defined over $T=n\cdot B$ stages
induce, for each player $k\in\mc{K}$, a utility that satisfies:
\begin{eqnarray}
\bigg|\gamma_k^T(\sigma^{\star},\tau^{\star}) - U_k^{\star} \bigg|
\leq  \epsilon, \qquad \forall k\in\mc{K}. \label{eq:UtiliteEpsilonU}
\end{eqnarray}
By repeating the strategies cyclically, we prove that there exists a $\bar{T}\geq\frac{N\cdot B}{\epsilon}$
such that for all $T'\geq \bar{T}$ and for all players $k\in \mc{K}$,
the expected $T'$ stage utility $ \gamma^{T'}(\sigma^{\star},\tau^{\star})$ is
$\epsilon$-closed of utility $U\in \conv u(\RR) \cap IR$.
Strategies $(\sigma^{\star},\tau^{\star} ) \in \Sigma \times \T$
satisfy the condition (i) of definition \ref{def:EqInfini}.

\begin{proof}[Lemma \ref{lemme:EvenementJeu}]
Denote $\hat{\bd{a}}_i^k(b)$ the sequence of actions of player $i\in\mc{K}$
observed by player $k\in\mc{K}$ over block $b\in\mc{B}$.
For all $\varepsilon>0$, there exists $n_1\in\N$ such that for all $n\geq n_1$, we have:
\begin{eqnarray}
& \PP\bigg( \bd{a}_i^n(b) \notin A_{\varepsilon}^{\star{n}}(\PP^{\star}_{i}) \bigg| \cap_{i,k\in\mc{K}}\bigg\{\bd{E}_k^i(b-1) = 0, \ldots, \bd{E}_k^i(b_1) = 0\bigg\}\bigg)\leq \varepsilon,& \forall i, k\in\mc{K},\;\forall b\in\mc{B},\nonumber \\ && \label{eq:DemoLemmeConv1}\\
& \PP\bigg(\hat{\bd{a}}_i^n(k,b) \neq  \bd{a}_i^n(b) \bigg| \cap_{i,k\in\mc{K}}\bigg\{\bd{E}_k^i(b-1) = 0, \ldots, \bd{E}_k^i(b_1) = 0\bigg\}\bigg)
\leq \varepsilon,& \forall i,k\in\mc{K},\;\forall b\in\mc{B}.\nonumber \\ &&\label{eq:DemoLemmeConv2}
\end{eqnarray}
Equations (\ref{eq:DemoLemmeConv1}) and (\ref{eq:DemoLemmeConv2})
come from the definition of strategies $(\sigma^{\star},\tau^{\star} ) \in \Sigma \times \T$.
When no deviation is detected, the players implement the main plan (Sec. \ref{sec:DemoPlanPrincipal}) by playing i.i.d. the mixed action
 $\PP^{\star} \in\Delta(\mc{A})$. \\
Equation (\ref{eq:DemoLemmeConv1}) is a consequence of Lemma \ref{lemma:TypicalSequencesCond} for the typical sequences and (\ref{eq:DemoLemmeConv2}) is a consequence of the coding result stated by Theorem \ref{theo:RobustSide} in Sec.
\ref{sec:subsec-resilient-source-coding} for an i.i.d. information source $\PP^{\star} \in \Delta(\mc{A})$.
More precisely, this inequality is a consequence of (\ref{eq:ErrorProbaKStrat}) that
guarantees at the beginning of block $b\in \mc{B}$, the players observe the sequence of actions played by the
other players over the block $b-2\in\mc{B}$ with probability $1- \varepsilon$. \\\\
Let us evaluate the probability of event $\bd{E}=1$.
\begin{tiny}
\begin{eqnarray}
\PP(\bd{E} = 1) &=&\PP\bigg(\cup_{b\in\mc{B},\atop i,k\in\mc{K}} \bigg\{\bd{E}_k^i(b) = 1\bigg\}\bigg) \label{eq:preuveConvergenceUtil1}\\
&\leq&   \sum_{b\in\mc{B}} \PP\bigg( \cup_{i,k\in\mc{K}} \bd{E}_k^i(b) = 1\bigg| \cap_{i,k\in\mc{K}}\bigg\{\bd{E}_k^i(b-1) = 0, \ldots, \bd{E}_k^i(b_1) = 0\bigg\}\bigg) \label{eq:preuveConvergenceUtil2}\\
&\leq&  \sum_{b\in\mc{B}, \atop i,k\in\mc{K}} \PP\bigg( \bigg\{\bd{a}_i^n(b) \notin A_{\varepsilon}^{\star{n}}(\PP^{\star}_{i})\bigg\}
\cup\bigg\{\hat{\bd{a}}_i^n(k,b) \neq  \bd{a}_i^n(b)\bigg\}\bigg| \cap_{i,k\in\mc{K}}\bigg\{\bd{E}_k^i(b-1) = 0, \ldots, \bd{E}_k^i(b_1) = 0\bigg\}\bigg)\nonumber\\&& \label{eq:preuveConvergenceUtil3}\\
&\leq&  \sum_{b\in\mc{B}, \atop i,k\in\mc{K}}
\PP\bigg( \bd{a}_i^n(b) \notin A_{\varepsilon}^{\star{n}}(\PP^{\star}_{i}) \bigg| \cap_{i,k\in\mc{K}}\bigg\{\bd{E}_k^i(b-1) = 0, \ldots, \bd{E}_k^i(b_1) = 0\bigg\}\bigg)\nonumber\\
&+& \sum_{b\in\mc{B}, \atop i,k\in\mc{K}}
\PP\bigg(\hat{\bd{a}}_i^n(k,b) \neq  \bd{a}_i^n(b) \bigg| \cap_{i,k\in\mc{K}}\bigg\{\bd{E}_k^i(b-1) = 0, \ldots, \bd{E}_k^i(b_1) = 0\bigg\}\bigg)\label{eq:preuveConvergenceUtil4}\\
&\leq&  2\varepsilon \cdot B \cdot K^2.  \label{eq:preuveConvergenceUtil5}
\end{eqnarray}
\end{tiny}
Equality (\ref{eq:preuveConvergenceUtil1}) comes from the definition of error event
provided by (\ref{eq:EvenementErreurJeu}).\\
Inequality (\ref{eq:preuveConvergenceUtil2}) comes from the property $\PP(A\cup B) = \PP(A) + \PP(B|A^c)\cdot \PP(A^c)$.\\
Inequalities (\ref{eq:preuveConvergenceUtil3}) and (\ref{eq:preuveConvergenceUtil4}) comes from the inequality of Boole.\\
Inequality (\ref{eq:preuveConvergenceUtil5}) comes from the inequalities (\ref{eq:DemoLemmeConv1}) and (\ref{eq:DemoLemmeConv2}).\\\\
As a conclusion, for all $\varepsilon>0$, there exists a $n_1\in\N$ such that for all
$n\geq n_1$, the condition (\ref{eq:LemmeEvenementJeu}) is satisfied.
\end{proof}

\begin{proof}[Lemma \ref{lemme:UtiliteEpsilon}]
When the event $\bd{E}=0$ occurs, then all the statistical tests
of players $\mc{K}$ at the beginning of each block
$b\in\mc{B}$ are negative (i.e.
$\bd{E}_k^i(b) = 0$, $\forall b\in\mc{B}$, $\forall i,k \in\mc{K}$).
In this case, the strategy $\tau^{\star}\in\T$ indicate
that the sequence of actions are generated with the same mixed strategy
 $\PP^{\star} \in \prod_{k\in \mc{K}}\Delta(\mc{A}_k) $ from stage to stage.
From the proof of Lemma \ref{lemme:EvenementJeu},
for all $\varepsilon>0$, there exists $n_1 \in \N$ such that for all
$n\geq n_1$, the sequences of block actions are typical with large probability.
More precisely, because $T = n \cdot B\geq n \geq n_1$, the sequences of actions $\bd{a}^T$
are typical with large probability.
\begin{eqnarray}
&&\PP\bigg(\bd{a}^n \in A_{\varepsilon}^{\star{n}}(\PP^{\star})\bigg| \bd{E}= 0  \bigg) \leq \varepsilon, \label{eq:SuiteTypiqueSachantE0}\\
&\Longrightarrow\qquad& \PP\bigg(\bd{a}^T \in A_{\varepsilon}^{\star{T}}(\PP^{\star})\bigg| \bd{E}= 0  \bigg) \leq \varepsilon. \label{eq:SuiteTypiqueSachantE0bis}
\end{eqnarray}
Suppose that $n\geq n_1$ defined from Lemma \ref{lemme:EvenementJeu}.
Recall the definition of the typical sequences and some implications thereof:
\begin{eqnarray}
&& {a}^T \in A_{\varepsilon}^{\star{T}}(\PP^{\star})\label{eq:DemoLemmeBorneTypiqueUtil1} \\
&\Longleftrightarrow& \sum_{a\in\mc{A}} \bigg| \frac{N(a | a^t)}{T} - \PP^{\star}(a) \bigg| \leq  \varepsilon\label{eq:DemoLemmeBorneTypiqueUtil2}\\
&\Longrightarrow& \sum_{a\in\mc{A}} \bigg| \frac{N(a | a^t)}{T}u_k(a) - \PP^{\star}(a)u_k(a) \bigg| \leq  \varepsilon \cdot \max_{a\in\mc{A}} |u_k(a)|\label{eq:DemoLemmeBorneTypiqueUtil3}\\
&\Longrightarrow&  \bigg| \sum_{a\in\mc{A}}\frac{N(a | a^t)}{T}u_k(a) - \sum_{a\in\mc{A}}\PP^{\star}(a)u_k(a) \bigg| \leq  \varepsilon \cdot \max_{a\in\mc{A}} |u_k(a)|\label{eq:DemoLemmeBorneTypiqueUtil4}\\
&\Longrightarrow&  \bigg| \frac{1}{T} \sum_{t = 1}^T u_k(a^t) - \E_{\PP^{\star}} \bigg[ u_k(\bd{a}) \bigg] \bigg| \leq  \varepsilon \cdot \max_{a\in\mc{A}} |u_k(a)|.  \label{eq:DemoLemmeBorneTypiqueUtil5}
\end{eqnarray}
Inequality (\ref{eq:DemoLemmeBorneTypiqueUtil2}) comes from the definition \ref{def:TypicalSequencesCond} of typical sequences.\\
Inequalities (\ref{eq:DemoLemmeBorneTypiqueUtil3}) come from the homogeneity property.\\
Equation (\ref{eq:DemoLemmeBorneTypiqueUtil4}) comes from the triangle inequality.\\
Equation (\ref{eq:DemoLemmeBorneTypiqueUtil5}) is a reformulation of (\ref{eq:DemoLemmeBorneTypiqueUtil4}) and allows us to obtain the following equations~:

\begin{eqnarray}
&&\bigg|\gamma_k^T(\sigma^{\star},\tau^{\star}) - \E_{\PP^{\star}} \bigg[ u_k(\bd{a}) \bigg]\bigg| \\
&= & \bigg|\sum_{a^T \in \mc{A}^T} \PP_{\sigma^{\star},\tau^{\star}}(a^T) \cdot  \frac{1}{T} \sum_{t = 1}^T u_k(a^t)  - \E_{\PP^{\star}} \bigg[ u_k(\bd{a}) \bigg]\bigg| \label{eq:DemoLemmeBorneUtil1}\\
&= & \bigg|\sum_{a^T \in A_{\varepsilon}^{\star{T}}(\PP^{\star})} \PP_{\sigma^{\star},\tau^{\star}}(a^T) \cdot  \frac{1}{T} \sum_{t = 1}^T u_k(a^t)
 + \sum_{{a}^T \notin A_{\varepsilon}^{\star{T}}(\PP^{\star})} \PP_{\sigma^{\star},\tau^{\star}}(a^T) \cdot  \frac{1}{T} \sum_{t = 1}^T u_k(a^t) - \E_{\PP^{\star}} \bigg[ u_k(\bd{a}) \bigg]\bigg| \nonumber \\ && \label{eq:DemoLemmeBorneUtil2}\\
&\leq & \sum_{a^T \in A_{\varepsilon}^{\star{T}}(\PP^{\star})} \PP_{\sigma^{\star},\tau^{\star}}(a^T) \cdot
\bigg|  \frac{1}{T} \sum_{t = 1}^T u_k(a^t)   -  \E_{\PP^{\star}} \bigg[ u_k(\bd{a}) \bigg]\bigg|
 + \sum_{{a}^T \notin A_{\varepsilon}^{\star{T}}(\PP^{\star})} \PP_{\sigma^{\star},\tau^{\star}}(a^T) \cdot  \frac{1}{T} \sum_{t = 1}^T |u_k(a^t)| \nonumber \\ &&\label{eq:DemoLemmeBorneUtil3}\\
&\leq & \sum_{a^T \in A_{\varepsilon}^{\star{T}}(\PP^{\star})} \PP_{\sigma^{\star},\tau^{\star}}(a^T)
\cdot \varepsilon \cdot \max_{a\in\mc{A}} |u_k(a)|
 + \sum_{{a}^T \notin A_{\varepsilon}^{\star{T}}(\PP^{\star})} \PP_{\sigma^{\star},\tau^{\star}}(a^T) \cdot \varepsilon \cdot \max_{a\in\mc{A}} |u_k(a)| \label{eq:DemoLemmeBorneUtil4} \\
 &\leq &  \varepsilon \cdot \max_{a\in\mc{A}} |u_k(a)| + \PP\bigg(\bd{a}^T \notin A_{\varepsilon}^{\star{T}}\bigg) \cdot  \varepsilon \cdot \max_{a\in\mc{A}} |u_k(a)| \label{eq:DemoLemmeBorneUtil5}\\
  &\leq &   \max_{a\in\mc{A}} |u_k(a)| \cdot \bigg( \varepsilon +
  \PP(\bd{a}^T \notin A_{\varepsilon}^{\star{T}}|\bd{E}=0 )\cdot \PP(\bd{E}=0 )
  +   \PP(\bd{a}^T \notin A_{\varepsilon}^{\star{T}}|\bd{E}=1 )\cdot \PP(\bd{E}=1 )
  \bigg)    \label{eq:DemoLemmeBorneUtil6}\\
    &\leq &   \max_{a\in\mc{A}} |u_k(a)| \cdot \bigg(\varepsilon +
  \PP(\bd{a}^T \notin A_{\varepsilon}^{\star{T}}|\bd{E}=0 ) +  \PP(\bd{E}=1 )
  \bigg)     \label{eq:DemoLemmeBorneUtil7}\\
      &\leq &   \max_{a\in\mc{A}} |u_k(a)| \cdot \bigg(2\varepsilon  +  \PP(\bd{E}=1 )
  \bigg)     \label{eq:DemoLemmeBorneUtil8}\\
        &\leq &   \max_{a\in\mc{A}} |u_k(a)| \cdot \bigg(2\varepsilon  +  2\varepsilon \cdot B \cdot K^2
  \bigg)      \label{eq:DemoLemmeBorneUtil9}\\
          &\leq &  4  \varepsilon \cdot \max_{a\in\mc{A}} |u_k(a)| \cdot  B\cdot K^2.       \label{eq:DemoLemmeBorneUtil10}
\end{eqnarray}
Equalities (\ref{eq:DemoLemmeBorneUtil1}) and (\ref{eq:DemoLemmeBorneUtil2}) come from the definition
 of the expected $T$-stages utility, see (\ref{eq:UtilInfini}).\\
Inequality (\ref{eq:DemoLemmeBorneUtil3}) comes from the triangle inequality.\\
Inequality (\ref{eq:DemoLemmeBorneUtil4}) comes from (\ref{eq:DemoLemmeBorneTypiqueUtil5}).\\
Inequalities (\ref{eq:DemoLemmeBorneUtil5}), (\ref{eq:DemoLemmeBorneUtil6}) and (\ref{eq:DemoLemmeBorneUtil7})
are reformulation of (\ref{eq:DemoLemmeBorneUtil4}).\\
Inequality (\ref{eq:DemoLemmeBorneUtil8}) comes from (\ref{eq:SuiteTypiqueSachantE0bis})
because  by assumption $T\geq n\geq n_1$.\\
Inequality (\ref{eq:DemoLemmeBorneUtil9}) comes from Lemma  \ref{lemme:EvenementJeu} because by assumption $n\geq n_1$.\\
Inequality (\ref{eq:DemoLemmeBorneUtil10}) is a reformulation with $B\cdot K^2 \geq 1$ which concludes the proof of Lemma \ref{lemme:UtiliteEpsilon}.\\\\
As a conclusion, for all $\varepsilon>0$, there exists $n_1\in\N$ such that for all
$n\geq n_1$, the condition (\ref{eq:UtiliteEpsilon}) is satisfied.
\end{proof}

\subsection{Condition (ii) of definition \ref{def:EqInfini} }\label{sec:DemoCondEquil}
In order to prove that the strategies $(\sigma^{\star}, \tau^{\star})\in \Sigma \times \T $
support a uniform equilibrium, we suppose that player $k\in \mc{K}$ implement
a deviating strategy  $\tau'_k \neq\tau^{\star}_k$ and we prove that
the deviation gain is less than $\epsilon>0$.\\\\
\subsubsection{First case: non-typical deviations}\label{sec:DemoCondEquil1Case}
Let $b - 1 \in\mc{B}$ the first action block over which the sequence of actions
$a_k(b-1) \notin A_{\varepsilon}^{\star{n}}(\PP^{\star}_k)$ of player $k\in\mc{K}$
is not typical. Denote $t_1(b)$ and $t_n(b)$ the indexes of the first and
the last stage of block $b\in\mc{B}$. \\\\

\textbf{Evaluate}, for player $k\in \mc{K}$, the utilities associated with the
strategies  $\tau^{\star}_k$ and $\tau'_k$.
\begin{eqnarray}
\gamma_k^T(\sigma^{\star}, \tau^{\star}_k,\tau^{\star}_{-k})
&=&   \E_{\sigma^{\star},\tau^{\star}_k,\tau^{\star}_{-k}} \bigg[ \frac{1}{T} \sum_{t=1}^T u_k(a^t) \bigg],\label{folktheoCondisigma} \\
\gamma_k^T(\sigma^{\star}, \tau'_k,\tau^{\star}_{-k})
&=&  \frac{1}{T} \E_{\sigma^{\star},\tau_k',\tau^{\star}_{-k}} \bigg[  \sum_{t=1}^T u_k(a^t) \bigg]\\
&=& \frac{1}{T}  \E_{\sigma^{\star},\tau_k',\tau^{\star}_{-k}}\bigg[
\sum_{t=1}^{t_n(b-2)} u_k(a^t) + \sum_{t=t_1(b-1)}^{t_n(b)} u_k(a^t) + \sum_{t=t_1(b+1)}^{T} u_k(a^t) \bigg].\nonumber \\ &&
\end{eqnarray}

\textbf{Approximation of the utility} associated with the strategies
$(\sigma^{\star},\tau^{\star})\in\Sigma \times \T$ between blocks $b+1\in\mc{B}$ and $B\in\mc{B}$.
\begin{lemma}\label{lemme:condEquil1}
Suppose that the encoder $\C$ and the players $\mc{K}$ follow the strategies
$(\sigma^{\star},\tau^{\star})\in\Sigma \times \mc{T}$.
Then for all $\epsilon>0$ and for all number of blocks $B\in\N$, there exists a block length $n_1$
such that for all $n \geq n_1 $, the following inequality is satisfied for all $1\leq b\leq B-1$:
\begin{eqnarray}
\E_{\sigma^{\star},\tau_k^{\star},\tau^{\star}_{-k}}\bigg[  \sum_{t=t_1(b+1)}^T u_k(a^t) \bigg] \geq n(B-b) \cdot (U_k^{\star} - \frac{\epsilon}{2}).\label{eq:LemmeCondEquil1}
\end{eqnarray}
\end{lemma}
\begin{proof}[Lemma \ref{lemme:condEquil1}]
Let us fix the parameter $\epsilon>0$ and suppose that the encoder $\C$
and the players $\mc{K}$ follows the strategies $(\sigma^{\star},\tau^{\star})\in\Sigma \times \mc{T}$.
This proof is built on Lemma \ref{lemme:UtiliteEpsilon} that prove for all
$\varepsilon>0$, there exists a block length $n_1\in \N$, such that for all
$n\geq n_1$, the expected utility satisfies the following equation~:
\begin{eqnarray}
\bigg|\gamma_k^T(\sigma^{\star},\tau^{\star}) - \E_{\PP^{\star}} \bigg[ u_k(\bd{a}_k,\bd{a}_{-k}) \bigg]\bigg|
\leq  4  \varepsilon \cdot \max_{a\in\mc{A}} |u_k(a)| \cdot  B\cdot K^2, \qquad \forall k\in\mc{K}. \label{eq:UtiliteEpsilon}
\end{eqnarray}
For a fixed number of blocks $B\in\N$, we choose $\varepsilon>0$ such that
$ \epsilon\geq 4  \varepsilon \cdot \max_{a\in\mc{A}} |u_k(a)| \cdot  B\cdot K^2)$.
Using the same reasoning as in Lemma  \ref{lemme:UtiliteEpsilon}, we prove that for all
$\varepsilon>0$, there exists a block length $n_1\in \N$, such that for all $n\geq n_1$,
the expected utility satisfies the following equation for all $k\in\mc{K}$:
\begin{footnotesize}
\begin{eqnarray}
&&\bigg|\E_{\sigma^{\star},\tau_k^{\star},\tau^{\star}_{-k}}\bigg[\frac{1}{n(B-b)}  \sum_{t=t_1(b+1)}^T u_k(a^t) \bigg]  - \E_{\PP^{\star}} \bigg[ u_k(\bd{a}_k,\bd{a}_{-k}) \bigg]\bigg|
\leq  4  \varepsilon \cdot \max_{a\in\mc{A}} |u_k(a)| \cdot  B\cdot K^2,\nonumber \\ &&\label{eq:UtiliteEpsilonB1}\\
&\Longrightarrow & \bigg|\E_{\sigma^{\star},\tau_k^{\star},\tau^{\star}_{-k}}\bigg[ \sum_{t=t_1(b+1)}^T u_k(a^t) \bigg]  - n(B-b)\cdot U_k^{\star}\bigg|
\leq  n(B-b)\cdot4  \varepsilon \cdot \max_{a\in\mc{A}} |u_k(a)| \cdot  B\cdot K^2,  \label{eq:UtiliteEpsilonB2}\\
&\Longrightarrow & \E_{\sigma^{\star},\tau_k^{\star},\tau^{\star}_{-k}}\bigg[ \sum_{t=t_1(b+1)}^T u_k(a^t) \bigg] \geq n(B-b)\cdot (U_k^{\star} - 4  \varepsilon \cdot \max_{a\in\mc{A}} |u_k(a)| \cdot  B\cdot K^2).  \label{eq:UtiliteEpsilonB3}
\end{eqnarray}
\end{footnotesize}
For a fixed block number $B\in\N$, we choose the parameter $\frac{\epsilon}{2} \geq 4  \varepsilon \cdot \max_{a\in\mc{A}} |u_k(a)| \cdot  B\cdot K^2 $ and the block length $n_1\in\N$ that satisfies (\ref{eq:DemoLemmeConv1}) and (\ref{eq:DemoLemmeConv2}) of Lemma \ref{lemme:EvenementJeu}.
We obtain the inequality (\ref{eq:LemmeCondEquil1}) of Lemma \ref{lemme:condEquil1}.
\end{proof}

\textbf{Upper bound on the deviation gain} obtained by player $k\in\mc{K}$
using the deviation strategy  $\tau_k'\in\T_k$, over blocks $b-1\in\mc{B}$ and $b\in\mc{B}$.
\begin{lemma}\label{lemme:condEquil2}
For every block $b\in \mc{B}$, the following inequality is satisfied:
\begin{eqnarray}
\E_{\sigma^{\star},\tau_k',\tau^{\star}_{-k}}\bigg[  \sum_{t=t_1(b-1)}^{t_n(b)} u_k(a^t) \bigg] &\leq&   2n \cdot\max_{a\in\mc{A}}|u_k(a)|, \qquad \forall k\in\mc{K}. \label{eq:LemmeCondEquil2}
 \end{eqnarray}
\end{lemma}
\begin{proof}[Lemma \ref{lemme:condEquil2}]
The proof is direct.
\begin{eqnarray}
\E_{\sigma^{\star},\tau_k',\tau^{\star}_{-k}}\bigg[  \sum_{t=t_1(b-1)}^{t_n(b)} u_k(a^t) \bigg]
&\leq & \E_{\sigma^{\star},\tau_k',\tau^{\star}_{-k}}\bigg[  (t_n(b) - t_1(b-1) + 1) \cdot\max_{a\in\mc{A}} |u_k(a)| \bigg]\\
&\leq&   2n \cdot\max_{a\in\mc{A}}|u_k(a)|, \qquad \forall k\in\mc{K}.
\end{eqnarray}
The deviation utility satisfies (\ref{eq:LemmeCondEquil2}).
\end{proof}

\textbf{Upper bound on the utility of player $k\in\mc{K}$ during the punishment phase}
induced by the prescribed strategies  $(\sigma^{\star},\tau^{\star})\in\Sigma \times \T$.
\begin{lemma}\label{lemme:condEquil3}
Suppose that $b-1\in\mc{B}$ is the first block on which the sequence of actions
$a_k(b-1) \notin A_{\varepsilon}^{\star{n}}(\PP^{\star}_k)$ of player $k\in\mc{K}$
is not typical. Suppose the block length satisfies $n\geq n_1$, defined
by (\ref{eq:DemoLemmeConv1}) and (\ref{eq:DemoLemmeConv2}).
The prescribed strategies $(\sigma^{\star},\tau^{\star})\in\Sigma \times \T$, induce the following equation:
\begin{eqnarray}
\E_{\sigma^{\star},\tau_k',\tau^{\star}_{-k}}\bigg[  \sum_{t=t_1(b+1)}^T u_k(a^t) \bigg]
& \leq& \varepsilon\cdot\max_{a\in\mc{A}}|u_k(a)| +   n(B - b) \cdot \upsilon_k. \label{eq:LemmeCondEquil3}
 \end{eqnarray}
\end{lemma}
Lemma \ref{lemme:condEquil3} is a consequence of the coding result stated by Theorem \ref{theo:RobustSide}.

\begin{proof}
We suppose that the actions $a_k(b-1) \notin A_{\varepsilon}^{\star{n}}(\PP^{\star}_k)$
of player $k\in\mc{K}$ over block $b-1\in\mc{B}$ are not typical and the length of the block
satisfies $n\geq n_1$, defined by (\ref{eq:DemoLemmeConv1}) and (\ref{eq:DemoLemmeConv2}).
Define the error event $\bd{E}_d$ related with the statistical test $\bd{E}_k^i(b+1)$
defined by (\ref{eq:TestStat}). $\bd{E}_d=0$ means that the statistical test
of all the players $j\neq k\in \mc{K}$ reveals that player $k\in\mc{K}$ deviates during block $b-1\in\mc{B}$.
\begin{eqnarray}
\bd{E}_d = \begin{cases}
0 \text{ if } \qquad \forall j\neq k,\quad \bd{E}_j^k(b+1) = 1,\\
1 \text{ otherwise. } \label{eq:EvenementErreurDev}
\end{cases}
\end{eqnarray}
The coding result stated by Theorem \ref{theo:RobustSide} allow us to bound
the probability of event $\bd{E}_d = 1$ knowing that
$a_k(b-1)\notin  A_{\varepsilon}^{\star{n}}(\PP^{\star}_k)$.
The following inequalities are valid for any deviation strategy
 $\tau_k'\in\T_k$ of player $k\in\mc{K}$.
\begin{eqnarray}
&&\PP_{\sigma^{\star},\tau_k',\tau^{\star}_{-k}}\bigg(\bd{E}_d = 1\bigg| a_k(b-1)\notin  A_{\varepsilon}^{\star{n}}(\PP^{\star}_k)\bigg ) \\
&=& \PP_{\sigma^{\star},\tau_k',\tau^{\star}_{-k}}\bigg(\exists j \neq k \in \mc{K},\quad \bd{E}_j^k(b+1)=0\bigg| a_k(b-1)\notin  A_{\varepsilon}^{\star{n}}(\PP^{\star}_k)\bigg)\label{eq:lemme3CondEquil1}\\
&\leq&\PP_{\sigma^{\star},\tau_k',\tau^{\star}_{-k}}\bigg(\cup_{j\neq k}\bigg\{\bd{E}_j^k(b+1)=0\bigg\}\bigg| a_k(b-1)\notin  A_{\varepsilon}^{\star{n}}(\PP^{\star}_k)\bigg)\label{eq:lemme3CondEquil2}\\
&\leq& \PP_{\sigma^{\star},\tau_k',\tau^{\star}_{-k}}\bigg(
\cup_{j\neq k}\bigg\{\hat{\bd{a}}_k^j(b-1)\in  A_{\varepsilon}^{\star{n}}(\PP^{\star}_k)\bigg\}\bigg| a_k(b-1)\notin  A_{\varepsilon}^{\star{n}}(\PP^{\star}_k)\bigg)\label{eq:lemme3CondEquil3}\\
&\leq& \PP_{e}(\lambda)\label{eq:lemme3CondEquil4}\\
&\leq&  \sum_{k\in\mc{K}} \max_{i \in \mc{K}}\max_{v_i \in  \Delta(\mc{A}_i^{\infty})} \PP(\bd{a}^n \neq \hat{\bd{a}}^n(k)| v_i)\label{eq:lemme3CondEquil5}\\
&\leq& \varepsilon, \qquad \forall \tau_k'\in\T_k.\label{eq:lemme3CondEquil6}
\end{eqnarray}
Inequalities \ref{eq:lemme3CondEquil1} and \ref{eq:lemme3CondEquil2} come from the definition $\bd{E}_d$ and Boole's inequality.\\
Inequality \ref{eq:lemme3CondEquil3} comes from the definition of the
statistical test \ref{eq:TestStat} presented section \ref{sec:DemoTestStat}.\\
Inequality \ref{eq:lemme3CondEquil4} come from the strategy of the encoder
 $\sigma^{\star}\in\Sigma$ and the decoding scheme of player $j\in\mc{K}$,
 described in sections \ref{sec:DemoStrategieCodage} and \ref{sec:DemoDecodage},
 based on the coding scheme $\lambda \in\Lambda(n)$ with $n\geq n_1$ during block $b-1\in\mc{B}$. \\
Inequality \ref{eq:lemme3CondEquil5} comes from the definition
of the error probability of the code $\lambda\in\Lambda(n)$.\\
Inequality \ref{eq:lemme3CondEquil6} comes from the coding result given by the Theorem
\ref{theo:RobustSide} for an arbitrarily varying information source (AVS).
When a player deviates, the actions are generated with an incertain
probability distribution satisfying the hypothesis  (\ref{Condition:SuiteDEtats}) and (\ref{eq:DistribArbitrary})
of the definition \ref{def:AVSunilateral}. We suppose the length of a block satisfies
$n\geq n_1$ defined by (\ref{eq:DemoLemmeConv1}) and (\ref{eq:DemoLemmeConv2}).
Therefore, from the Theorem \ref{theo:RobustSide}, there exists a code $\lambda \in\Lambda(n)$ for which the error
probability $\PP_e(\lambda)$ is bounded by $\varepsilon>0$, for every unilateral deviation
$\tau_k' \in\T_k$ of player $k\in\mc{K}$. This inequality allow us to obtain an upper bound on the utility
of player  $k\in\mc{K}$ during the punishment phase stated by the Lemma \ref{lemme:condEquil3}.
\begin{eqnarray}
&&\E_{\sigma^{\star},\tau_k',\tau^{\star}_{-k}}\bigg[  \sum_{t=t_1(b+1)}^T u_k(a^t) \bigg]\label{eq:lemme3CondFinale1}\\
&=& \sum_{a_{t_1(b+1)}^{T} \in \mc{A}^{t_n(b)}}
\cdot \PP_{\sigma^{\star},\tau_k',\tau^{\star}_{-k}}\bigg(a_{t_1(b+1)}^{T},\bd{E}_d = 1\bigg| a_k(b-1)\notin  A_{\varepsilon}^{\star{n}}(\PP^{\star}_k)\bigg)
\cdot\bigg[  \sum_{t=t_1(b+1)}^T u_k(a^t) \bigg]\nonumber\\
&+& \sum_{a_{t_1(b+1)}^{T} \in \mc{A}^{t_n(b)}}
\cdot \PP_{\sigma^{\star},\tau_k',\tau^{\star}_{-k}}\bigg(a_{t_1(b+1)}^{T},\bd{E}_d = 0\bigg|a_k(b-1)\notin  A_{\varepsilon}^{\star{n}}(\PP^{\star}_k) \bigg)
\cdot\bigg[  \sum_{t=t_1(b+1)}^T u_k(a^t) \bigg]\nonumber \\ &&\label{eq:lemme3CondFinale2}\\
&\leq& \PP_{\sigma^{\star},\tau_k',\tau^{\star}_{-k}}\bigg(\bd{E}_d = 1\bigg| a_k(b-1)\notin  A_{\varepsilon}^{\star{n}}(\PP^{\star}_k)
\bigg) \cdot \max_{a\in\mc{A}} |u_k(a)|  + \bigg[  n(B - b) \cdot  \upsilon_k \bigg]\label{eq:lemme3CondFinale3}\\
&\leq& \varepsilon \cdot \max_{a\in\mc{A}} |u_k(a)|
+ \bigg[  n(B - b) \cdot  \upsilon_k \bigg].\label{eq:lemme3CondFinale4}
\end{eqnarray}
Inequality \ref{eq:lemme3CondFinale2} comes from the definition of the expectation \ref{eq:lemme3CondFinale1} knowing that
the sequence of actions  $a_k(b-1)\notin  A_{\varepsilon}^{\star{n}}(\PP^{\star}_k)$ is not typical
over the block $b-1 \in \mc{B}$.\\
Inequality \ref{eq:lemme3CondFinale3} comes from the punishment plan stated by the strategy
$\tau^{\star}_{-k}$ when all the players $j\neq k$ detect the deviation of player
$k\in\mc{K}$ ($\bd{E}_d=1$). For each stage $t_1(b+1)\leq t\leq T$, the utility of player $k\in\mc{K}$
is less than the min-max level $u_k(a^t)\leq \upsilon_k$.\\
Inequality \ref{eq:lemme3CondFinale4} comes from (\ref{eq:lemme3CondEquil3}).
\end{proof}

\textbf{Equilibrium condition.} Hypothesis (\ref{eq:DemoConditionB}) over the number of blocks
$B\in \R$ and the results of Lemma \ref{lemme:condEquil1}, \ref{lemme:condEquil2} and
\ref{lemme:condEquil3} allow us to obtain the following inequalities:
\begin{tiny}
\begin{eqnarray}
&B&\geq \frac{8 \cdot\max_{a\in\mc{A}} |u_k(a)|}{\epsilon}\label{eq:condEquilFinal1}\\
\Longrightarrow & - \frac{B\epsilon}{2} +  B\epsilon  &\geq 4 \cdot\max_{a\in\mc{A}} |u_k(a)|\label{eq:condEquilFinal2}\\
\Longrightarrow & (B - b ) \cdot U_k^{\star} - \frac{B\epsilon}{2} +  B\epsilon  &\geq 4 \cdot\max_{a\in\mc{A}} |u_k(a)| + (B - b ) \cdot \upsilon_k\label{eq:condEquilFinal3}\\
\Longrightarrow & (B - b + 2) \cdot U_k^{\star} - \frac{(B - b)\epsilon}{2} +  B\epsilon  &\geq 2 \cdot\max_{a\in\mc{A}} |u_k(a)| + \epsilon + \varepsilon\cdot\max_{a\in\mc{A}} |u_k(a)| + (B - b ) \cdot \upsilon_k\label{eq:condEquilFinal4}\\
\Longrightarrow & n(B - b + 2) \cdot U_k^{\star} - n\frac{(B - b + 2)\epsilon}{2} +  nB\epsilon  &\geq 2 n\cdot\max_{a\in\mc{A}} |u_k(a)|  + n\varepsilon\cdot\max_{a\in\mc{A}} |u_k(a)| + n(B - b ) \cdot \upsilon_k\label{eq:condEquilFinal5}\\
\Longrightarrow & \E_{\sigma^{\star},\tau_k^{\star},\tau^{\star}_{-k}}\bigg[  \sum_{t=t_1(b-1)}^T u_k(a^t) \bigg] +  nB\epsilon  &\geq 2 n\cdot\max_{a\in\mc{A}} |u_k(a)|  + n\varepsilon\cdot\max_{a\in\mc{A}} |u_k(a)| + n(B - b ) \cdot \upsilon_k\label{eq:condEquilFinal6}\\
\Longrightarrow & \E_{\sigma^{\star},\tau_k^{\star},\tau^{\star}_{-k}}\bigg[  \sum_{t=t_1(b-1)}^T u_k(a^t) \bigg] +  nB\epsilon  &\geq
\E_{\sigma^{\star},\tau_k',\tau^{\star}_{-k}}\bigg[  \sum_{t=t_1(b-1)}^{t_n(b)} u_k(a^t) \bigg]  + n\varepsilon\cdot\max_{a\in\mc{A}} |u_k(a)| + n(B - b ) \cdot \upsilon_k \nonumber\\&&\label{eq:condEquilFinal7}\\
\Longrightarrow & \E_{\sigma^{\star},\tau_k^{\star},\tau^{\star}_{-k}}\bigg[  \sum_{t=t_1(b-1)}^T u_k(a^t) \bigg] +  nB\epsilon  &\geq
\E_{\sigma^{\star},\tau_k',\tau^{\star}_{-k}}\bigg[  \sum_{t=t_1(b-1)}^{t_n(b)} u_k(a^t) \bigg]
+ \E_{\sigma^{\star},\tau_k',\tau^{\star}_{-k}}\bigg[  \sum_{t=t_1(b+1)}^T u_k(a^t) \bigg] \nonumber\\&&\label{eq:condEquilFinal8}\\
\Longrightarrow &\E_{\sigma^{\star},\tau_k^{\star},\tau^{\star}_{-k}}\bigg[ \sum_{t=1}^{T} u_k(a^t)  \bigg] + T\cdot \varepsilon  &\geq \E_{\sigma^{\star},\tau_k',\tau^{\star}_{-k}}\bigg[ \sum_{t=1}^{t_n(b-2)} u_k(a^t) +  \sum_{t=t_1(b-1)}^{t_n(b)} u_k(a^t) +  \sum_{t=t_1(b+1)}^T u_k(a^t) \bigg]\nonumber\\&&\label{eq:condEquilFinal9}\\
\Longrightarrow & \gamma_k^T(\sigma^{\star}, \tau^{\star}_k,\tau^{\star}_{-k}) + \varepsilon &\geq \gamma_k^T(\sigma^{\star}, \tau'_k,\tau^{\star}_{-k}).\label{eq:condEquilFinal10}
\end{eqnarray}
\end{tiny}
Inequality \ref{eq:condEquilFinal1} comes from the hypothesis (\ref{eq:DemoConditionB}) over the number of blocks $B\in \R$.\\
Inequality \ref{eq:condEquilFinal2} comes from the reformulation of inequality (\ref{eq:condEquilFinal1}).\\
Inequality \ref{eq:condEquilFinal3} comes from the hypothesis  of individual rationality $U_k^{\star} \geq \upsilon_k$ stated by the definition  \ref{def:minmaxIRreal}.\\
Inequalities \ref{eq:condEquilFinal4} and \ref{eq:condEquilFinal5} come from the reformulation of inequality (\ref{eq:condEquilFinal3}) with
 $\epsilon\leq \max_{a\in\mc{A}} |u_k(a)|$ and $\varepsilon \leq 1$.\\
Inequality \ref{eq:condEquilFinal6} comes from Lemma \ref{lemme:condEquil1} which provides an approximation of the utility associated with the strategies $(\sigma^{\star},\tau^{\star})\in\Sigma \times \T$.\\
Inequality \ref{eq:condEquilFinal7} comes from the Lemma \ref{lemme:condEquil2} which provides an upper bound on the deviation gain obtained by player  $k\in\mc{K}$ while playing the strategy $\tau_k'\in\T_k$.\\
Inequality \ref{eq:condEquilFinal8} comes from Lemma \ref{lemme:condEquil3} which is a consequence of the coding result stated by Theorem \ref{theo:RobustSide}. This result provides an upper bound over the utility of player $k\in\mc{K}$ during the punishment phase.\\
Inequality \ref{eq:condEquilFinal9} comes from the fact that $b-1 \in\mc{B}$ is the first block on which the action sequence
 $a_k(b-1) \notin A_{\varepsilon}^{\star{n}}(\PP^{\star}_k)$ of player $k\in\mc{K}$ is not typical.\\
Inequality \ref{eq:condEquilFinal10} comes from the definition of the utilities of the $T$-stages repeated game stated equation (\ref{eq:UtilInfini}).\\
We prove the strategies $(\sigma^{\star},\tau^{\star}) \in \Sigma \times \T$
satisfy the equilibrium condition stated by point (ii) of the definition \ref{def:EqInfini}.\\

\subsubsection{Second case: typical deviations}\label{sec:DemoCondEquil2Case}
Let us fix $\epsilon>0$ and suppose player $k \in\mc{K}$ uses a deviating strategy $\tau_k'\in\T_k$
such that the action sequence of player $k\in\mc{K}$ over each block  $b\in\mc{B}$
belong to the set of typical sequences $a_k(b) \in A_{\varepsilon}^{\star{n}}(\PP^{\star}_k)$.
From the proof of Lemma \ref{lemme:UtiliteEpsilon}, for all $\varepsilon>0$,
there exists $n_1\in\N$ such that for all $n\geq n_1$ we have the following implication (\ref{eq:DemoLemmeDevTypique2}).
Taking $\varepsilon \cdot \max_{a\in\mc{A}} |u_k(a)| \leq \epsilon$, we obtain the implication (\ref{eq:DemoLemmeDevTypique3}).
\begin{eqnarray}
&& {a}^T \in A_{\varepsilon}^{\star{T}}(\PP^{\star})\label{eq:DemoLemmeDevTypique1} \\
&\Longrightarrow&  \bigg| \frac{1}{T} \sum_{t = 1}^T u_k(a^t) - \E_{\PP^{\star}} \bigg[ u_k(\bd{a}) \bigg] \bigg| \leq  \varepsilon \cdot \max_{a\in\mc{A}} |u_k(a)|.  \label{eq:DemoLemmeDevTypique2}\\
&\Longrightarrow&  \gamma_k^T(\sigma^{\star}, \tau'_k,\tau^{\star}_{-k}) \leq \gamma_k^T(\sigma^{\star}, \tau^{\star}_k,\tau^{\star}_{-k}) + \epsilon.\label{eq:DemoLemmeDevTypique3}
\end{eqnarray}
The utility provided by the strategies $(\sigma^{\star},\tau^{\star}) \in \Sigma \times \T$
satisfies the equilibrium condition stated by point (ii) in definition \ref{def:EqInfini}.\\

\subsection{Conclusion}\label{sec:DemoConclusionThJeu}
We showed that by setting the parameter $\epsilon>0$, we obtain a condition on the number of blocks $B\in\N$ given by (\ref{eq:condEquilFinal1}), then a condition over the coding parameter $\varepsilon>0 $
given by (\ref{eq:UtiliteEpsilon}) and then a condition over the block length $n\geq n_1 $
given by (\ref{eq:DemoLemmeConv1}) and (\ref{eq:DemoLemmeConv2}).\\
For all $U\in u(\RR) \cap IR$, these parameters allow us to construct
a pair of strategies $(\sigma^{\star},\tau^{\star}) \in \Sigma \times \T$
over $T=n\cdot B$ stages that satisfies the conditions (\ref{eq:DemoCondition1}) and (\ref{eq:DemoCondition2}).
By repeating these strategies cyclically, we show that any vector of utility
$U \in \conv u(\RR) \cap IR$ satisfy both conditions (i) and (ii) (i.e. definition  \ref{def:EqInfini})
of the uniform equilibrium. The utility $U \in \conv u(\RR) \cap IR$ is a uniform
equilibrium utility for the infinite repeated game  $\Gamma^{\infty}$.

\section{Review of typical sequences}\label{sec:AppTypicaSeq}

The achievability part of the coding theorems are based on the properties of the typical sequences.
This section provides some recall on this notions that can also be found in \cite{cover-book-2006} and \cite{CsiszarKorner(Book)81}.
\begin{definition}[Typical sequences  \cite{CsiszarKorner(Book)81}]\label{def:TypicalSequencesCond}
Let $\QQ\in \Delta(\mc{X}\times \mc{Y})$ a probability distribution over $\mc{X}\times \mc{Y}$. The typical sequences and the conditional typical sequences are defined as follows:
\begin{footnotesize}
 \begin{eqnarray}
  A_{\varepsilon}^{n{\star}}(\mc{X})&=&\left\{x^n\in \mc{X}^n;\;  \sum_{x\in \mc{X}}\left|\frac{N(x|x^n)}{n}- \QQ(x)\right|\leq \varepsilon,\;
 \forall x\in \mc{X},\;\QQ(x)=0\Longrightarrow N(x|x^n)=0
  \right\}.\nonumber \\
 A_{\varepsilon}^{n{\star}}(\mc{X}|y^n)&=&\bigg\{x^n\in \mc{X}^n;\;  \sum_{x\in \mc{X}, \atop y \in \mc{Y}}\left|\frac{N(x,y|x^n,y^n)}{n}- \QQ(x,y)\right|\leq \varepsilon, \nonumber \\
 &&\forall (x,y)\in \mc{X} \times \mc{Y},\;\QQ(x,y)=0\Longrightarrow N(x,y|x^n,y^n)=0
  \bigg\}.
 \end{eqnarray}
 \end{footnotesize}
\end{definition}

\begin{lemma}[Properties of the typical sequences \cite{CsiszarKorner(Book)81}]\label{lemma:TypicalSequencesCond}
Let  $\QQ\in \Delta(\mc{X}\times \mc{Y})$ a probability distribution, $\QQ^{\otimes n}$ a $n$-product of the probability distribution and  $y^n\in A_{\varepsilon}^{n{\star}}(\mc{Y})$ a typical sequence. For all $\varepsilon>0$, there exists $n\in\N$ such that:
\begin{eqnarray}
1&=& \QQ^{\otimes n}\bigg(\bd{x}^n\in A_{\varepsilon}^{n{\star}}(\mc{X})\bigg)   \label{eq:ProbTypicity},\\
1 &=& \QQ^{\otimes n}\bigg(\bd{x}^n\in A_{\varepsilon}^{n{\star}}(\mc{X}|y^n)\bigg|y^n\bigg)   \label{eq:ProbCondTypicity},\\
2^{n(H(\bd{x}) - c\varepsilon)}&\leq& |A_{\varepsilon}^{n{\star}}(\mc{X})| \leq 2^{n(H(\bd{x}) + c\varepsilon)},\\
2^{n(H(\bd{x}|\bd{y}) - c\varepsilon)}&\leq& |A_{\varepsilon}^{n{\star}}(\mc{X}|y^n) | \leq 2^{n(H(\bd{x}|\bd{y}) + c\varepsilon)}.
\end{eqnarray}
where $c= \log \bigg(\max_{x\in \mc{X}}\frac{1}{\QQ(x)}\bigg)$ is constant.
\end{lemma}
This result states that an i.i.d. sequence of symbols is almost surely typical when $n$ goes to $+\infty$
and it provides an upper and a lower bound on the size of the sets of typical sequences.

\begin{lemma}[Packing Lemma \cite{elgamal-it-1981}]\label{lemma:MutualProb}
Let $\QQ\in \Delta(\mc{U}\times \mc{V})$ a correlated probability distribution, $\QQ_U$ (resp. $\QQ_V$) the marginal induced by $\QQ$ over $\mc{U}$ (resp. $\mc{V}$), $\QQ_U^{\otimes n}$ and $\QQ_V^{\otimes n}$ the $n$-product of marginal probability. Let $R_I$ and $R_J$ real numbers,\\
\begin{itemize}
\item[$\bullet$] $(u_i^n)_{i\in \{1,\ldots,2^{nR_I}\}} \in \mc{U}^n$ a family of sequences drawn with $\QQ_U^{\otimes n}$,
\item[$\bullet$] $(v_j^n)_{j\in \{1,\ldots,2^{nR_J}\}} \in \mc{V}^n$ a family of sequences drawn with $\QQ_V^{\otimes n}$,
\end{itemize}
If the condition (\ref{eq:mutualcovering3}) is satisfied,\\
\begin{eqnarray}
R_I +R_J &<& I_{\QQ}(\bd{u};\bd{v}),
\label{eq:mutualcovering3}\\ \nonumber
\end{eqnarray}
then for all $\varepsilon>0$, there exists a $\bar{n}\geq 0$ such that for all $n\geq\bar{n}$,\\
\begin{eqnarray}
\PP\bigg(\cup_{i\in \mc{I},\atop j\in \mc{J}}\bigg\{(u_i^n,v_j^n)\in A_{\varepsilon}^{{\star}n}(\mc{U}\times \mc{V} )\bigg\} \bigg)\leq \varepsilon.\label{eq:mutualcovering1}\\ \nonumber
\end{eqnarray}
Where $I_{\QQ}(\bd{u};\bd{v})$ denote the mutual information \cite{cover-book-2006} with respect to the probability distribution $\QQ$.
\end{lemma}

\nocite{Witsenhausen76}
\bibliographystyle{spbasic}
\bibliography{BiblioMael}

\end{document}